\documentclass{article}
\usepackage{silence}
\usepackage{style/foci}
\usepackage{tgpagella}   
\usepackage{tgheros}     

\usepackage[utf8]{inputenc}
\usepackage[T1]{fontenc}
\usepackage{microtype}

\usepackage{amsmath}
\usepackage{amssymb}
\usepackage{amsfonts}
\usepackage{amsthm}
\usepackage{mathtools}
\usepackage{mathrsfs}
\usepackage{physics}
\usepackage{braket}
\usepackage{slashed}
\usepackage{nicefrac}
\usepackage{textcomp}
\usepackage{dsfont}
\usepackage{bbm}
\usepackage{bm}

\usepackage{graphicx}
\usepackage{subcaption}
\usepackage[export]{adjustbox}
\usepackage{float}
\usepackage{booktabs}
\usepackage{dcolumn}
\newcolumntype{d}[1]{D{.}{.}{#1}}
\usepackage{bigstrut, tabularx, multirow, makecell, diagbox}
\usepackage{colortbl}
\usepackage{tabularray}
\UseTblrLibrary{booktabs}
\usepackage{threeparttable}
\usepackage{tablefootnote}
\usepackage{fontawesome5}

\usepackage{placeins}
\usepackage{caption}
\usepackage{footnote}
\usepackage{enumitem}
\usepackage{multicol}
\usepackage{xspace}
\usepackage{titletoc}
\usepackage{titlesec}
\usepackage[bottom]{footmisc}
\usepackage{setspace}

\usepackage{wrapfig}
\usepackage{tikz}
\usepackage{quantikz}
\usepackage{dashbox}
\usepackage{mdframed}
\usepackage{marvosym}
\usepackage{pifont}
\usepackage{CJK}
\usepackage{url}

\usepackage[table,x11names]{xcolor}
\usepackage[most]{tcolorbox}
\tcbuselibrary{breakable}
\usetikzlibrary{decorations.pathreplacing, fit}

\definecolor{primaryblue}{HTML}{0066CC}
\definecolor{accentcyan}{HTML}{00D4AA}
\definecolor{warmorange}{HTML}{FF6B35}
\definecolor{deepgray}{HTML}{2C3E50}
\definecolor{lightgray}{HTML}{F8F9FA}
\definecolor{gradientstart}{HTML}{667eea}
\definecolor{gradientend}{HTML}{764ba2}

\definecolor{citecolor}{HTML}{0071bc}
\definecolor{citeblue}{RGB}{0, 113, 188}
\definecolor{linkcolor}{HTML}{9A4D92}
\definecolor{firebrick}{rgb}{0.698,0.133,0.133}

\definecolor{paleviolet}{HTML}{E1EEFC}
\definecolor{CarolinaUltraLight}{HTML}{E7F4FC}
\definecolor{lightgrey}{RGB}{247, 247, 247}
\definecolor{shadecolor}{HTML}{EFEFEF}
\definecolor{lightyellow}{rgb}{1.0, 0.95, 0.7}
\definecolor{lightblue}{rgb}{0.90, 0.95, 1.0}
\definecolor{light-gray}{gray}{0.95}

\definecolor{darkgrey}{rgb}{0.5, 0.5, 0.5}
\definecolor{darkgreen}{rgb}{0, 0.5, 0}
\definecolor{mydarkblue}{rgb}{0,0.08,0.45}
\definecolor{mydarkblue2}{rgb}{0.133, 0.133, 0.698}
\definecolor{echodrk}{HTML}{0099cc}
\definecolor{mymauve}{rgb}{0.58,0,0.82}
\definecolor{midnightblue}{rgb}{0.1,0.1,0.44}
\definecolor{oxfordblue}{rgb}{0.0,0.13,0.28}
\definecolor{prussianblue}{rgb}{0.0,0.19,0.33}
\definecolor{coolteal}{rgb}{0, 0.45, 0.45}
\definecolor{olive}{rgb}{0.1, 0.3, 0}
\definecolor{mypurple}{rgb}{0.5,0,0.5}
\definecolor{almond}{rgb}{0.94, 0.87, 0.8}

\definecolor{blue_ampEncoding}{HTML}{DAE8FC}
\definecolor{green_encoder}{HTML}{D5E8D4}
\definecolor{purple_decoder}{HTML}{E1D5E7}
\definecolor{yellow_measure}{HTML}{FFF2CC}
\definecolor{gray_block}{HTML}{F5F5F5}
\definecolor{pink_dru}{HTML}{FAD9D5}
\definecolor{orange_v}{HTML}{FAD7AC}

\definecolor{colorA}{rgb}{1,0,0}
\definecolor{colorB}{rgb}{0,0.3,1}
\definecolor{colorC}{rgb}{0.9,0.8,0.2}
\definecolor{colorD}{rgb}{0,0.65,0}
\definecolor{lesslightgray}{rgb}{0.5,0.5,0.5}
\definecolor{fundamental}{RGB}{55, 110, 111}
\definecolor{Gred}{RGB}{219, 50, 54}
\definecolor{ToCgreen}{RGB}{0, 128, 0}
\definecolor{Sepia}{RGB}{112, 66, 20}
\definecolor{Dblue}{rgb}{0,0.08,0.45}
\definecolor{Blue}{rgb}{0, 0, 0.8}
\definecolor{blue}{rgb}{0,0,1}
\definecolor{UNCblue!10}{rgb}{0.84,0.91,0.98}
\definecolor{RowAlt}{rgb}{0.98,0.98,0.99}

\definecolor{CarolinaBlue}{HTML}{7BAFD4}        
\definecolor{CarolinaLightBlue}{HTML}{B3D4E5}   
\definecolor{CarolinaUltraLight}{HTML}{E8F4F8}  
\definecolor{CarolinaText}{HTML}{1C2B33}        

\usepackage[pagebackref=true,breaklinks=true,colorlinks,hyperfootnotes=false]{hyperref}
\hypersetup{
  colorlinks,
  citecolor=citeblue,
  linkcolor=firebrick,
  urlcolor=firebrick
}
\usepackage[nameinlink,capitalize,noabbrev]{cleveref}

\titlespacing\section{0pt}{4pt plus 4pt minus 2pt}{-2pt plus 2pt minus 2pt}
\titlespacing\subsection{0pt}{2pt plus 4pt minus 2pt}{-2pt plus 2pt minus 2pt}
\titlespacing\subsubsection{0pt}{2pt plus 4pt minus 2pt}{-2pt plus 2pt minus 2pt}

\makeatletter
\def\th@remark{%
  \thm@headfont{\bfseries}%
  \normalfont 
  \thm@preskip\topsep \divide\thm@preskip\tw@
  \thm@postskip\thm@preskip
}
\makeatother

\theoremstyle{definition}

\tcbset{focithm/.style={
  enhanced, breakable, frame hidden,
  colback=black!5,
  width=\linewidth,
  enlarge left by=0pt, enlarge right by=0pt,
  boxsep=6pt, boxrule=0pt,
  left=4pt, right=4pt, top=1pt, bottom=1pt,
  arc=8pt,
  before skip=\topsep, after skip=\topsep
}}

\tcolorboxenvironment{theorem}{focithm}

\tcolorboxenvironment{lemma}{focithm}

\tcolorboxenvironment{corollary}{focithm}

\newtheorem{proposition}{Proposition}[section]
\tcolorboxenvironment{proposition}{focithm}

\tcolorboxenvironment{definition}{focithm}

\newtheorem{remark}{Remark}[section]

\tcolorboxenvironment{assumption}{focithm}

\tcolorboxenvironment{claim}{focithm}

\tcolorboxenvironment{problem}{focithm}

\tcolorboxenvironment{question}{focithm}

\newtcolorbox{titleblock}{
  enhanced,
  frame hidden,
  colback=CarolinaUltraLight,
  colframe=CarolinaUltraLight,
  boxrule=0pt,
  arc=10pt,
  left=14pt,
  right=14pt,
  top=14pt,
  bottom=14pt,
  width=\linewidth,
  before skip=12pt plus 4pt,
  after skip=12pt plus 4pt,
  grow to left by=1.5pt,
  grow to right by=1.5pt,
  before upper={
    \setlength{\parindent}{0cm}
    \setlength{\parskip}{0.5cm}
  }
}

\crefname{theorem}{Theorem}{Theorems}
\crefname{proposition}{Proposition}{Propositions}
\crefname{lemma}{Lemma}{Lemmas}
\crefname{corollary}{Corollary}{Corollaries}
\crefname{definition}{Definition}{Definitions}
\crefname{assumption}{Assumption}{Assumptions}
\crefname{remark}{Remark}{Remarks}
\crefname{problem}{Problem}{Problems}
\crefname{property}{Property}{property}
\crefname{question}{Question}{Questions}

\numberwithin{equation}{section}
\numberwithin{theorem}{section}
\numberwithin{proposition}{section}
\numberwithin{definition}{section}
\numberwithin{lemma}{section}
\numberwithin{assumption}{section}
\numberwithin{remark}{section}

\def\1{\bm{1}}

\makeatletter
\let\save@mathaccent\mathaccent
\newcommand*\if@single[3]{%
    \setbox0\hbox{${\mathaccent"0362{#1}}^H$}%
    \setbox2\hbox{${\mathaccent"0362{\kern0pt#1}}^H$}%
    \ifdim\ht0=\ht2 #3\else #2\fi
}
\newcommand*\rel@kern[1]{\kern#1\dimexpr\macc@kerna}
\newcommand*\widebar[1]{\@ifnextchar^{{\wide@bar{#1}{0}}}{\wide@bar{#1}{1}}}
\newcommand*\wide@bar[2]{\if@single{#1}{\wide@bar@{#1}{#2}{1}}{\wide@bar@{#1}{#2}{2}}}
\newcommand*\wide@bar@[3]{%
    \begingroup
    \def\mathaccent##1##2{%
        \let\mathaccent\save@mathaccent
        \if#32 \let\macc@nucleus\first@char \fi
        \setbox\z@\hbox{$\macc@style{\macc@nucleus}_{}$}%
        \setbox\tw@\hbox{$\macc@style{\macc@nucleus}{}_{}$}%
        \dimen@\wd\tw@
        \advance\dimen@-\wd\z@
        \divide\dimen@ 3
        \@tempdima\wd\tw@
        \advance\@tempdima-\scriptspace
        \divide\@tempdima 10
        \advance\dimen@-\@tempdima
        \ifdim\dimen@>\z@ \dimen@0pt\fi
        \rel@kern{0.6}\kern-\dimen@
        \if#31
        \overline{\rel@kern{-0.6}\kern\dimen@\macc@nucleus\rel@kern{0.4}\kern\dimen@}%
        \advance\dimen@0.4\dimexpr\macc@kerna
        \let\final@kern#2%
        \ifdim\dimen@<\z@ \let\final@kern1\fi
        \if\final@kern1 \kern-\dimen@\fi
        \else
        \overline{\rel@kern{-0.6}\kern\dimen@#1}%
        \fi
    }%
    \macc@depth\@ne
    \let\math@bgroup\@empty \let\math@egroup\macc@set@skewchar
    \mathsurround\z@ \frozen@everymath{\mathgroup\macc@group\relax}%
    \macc@set@skewchar\relax
    \let\mathaccentV\macc@nested@a
    \if#31
    \macc@nested@a\relax111{#1}%
    \else
    \def\gobble@till@marker##1\endmarker{}%
    \futurelet\first@char\gobble@till@marker#1\endmarker
    \ifcat\noexpand\first@char A\else
    \def\first@char{}%
    \fi
    \macc@nested@a\relax111{\first@char}%
    \fi
    \endgroup
    }
\makeatother

\DeclareMathAlphabet{\mathsfit}{\encodingdefault}{\sfdefault}{m}{sl}
\SetMathAlphabet{\mathsfit}{bold}{\encodingdefault}{\sfdefault}{bx}{n}

\let\tilde\widetilde
\let\hat\widehat

\renewcommand{\arraystretch}{1.15}
\makeatletter

\def\foci@authors{}
\def\foci@affil{}
\def\foci@email{}
\def\foci@project{}
\def\foci@date{\today}
\def\foci@authorsep{}

\newcommand{\fociauthor}[2]{%
  \g@addto@macro\foci@authors{\foci@authorsep\mbox{#1\textsuperscript{#2}}\gdef\foci@authorsep{,\ }}}
\newcommand{\fociaffil}[1]{\gdef\foci@affil{#1}}
\newcommand{\fociemail}[1]{\gdef\foci@email{#1}}
\def\foci@corremail{\foci@email}
\newcommand{\focicorremail}[1]{\gdef\foci@corremail{#1}}
\newcommand{\fociproject}[1]{\gdef\foci@project{#1}}
\newcommand{\focidate}[1]{\gdef\foci@date{#1}}
\newcommand{\fociauthorlist}{\gdef\foci@authorsep{}\foci@authors}
\newcommand{\focicorr}{\textrm{\Letter}}

\def\foci@blfootnote{\gdef\@thefnmark{}\@footnotetext}
\newcommand{\focifootnotes}{%
  \foci@blfootnote{
    \textsuperscript{\Letter}Corresponding authors: \texttt{\foci@corremail}}}

\newlength{\foci@mw}
\newcommand{\focidrawmark}[3]{%
  \setlength{\foci@mw}{#2}%
  \begin{scope}[shift={#3}, x={\foci@mw/564.03}, y={-\foci@mw/564.03}]
    \foreach \cx/\cy in {282.02/66.03, 282.02/222.08, 215.77/348.8, 348.73/348.8,
                         414.88/474.93, 149.16/474.93, 67.06/222.08, 496.99/222.08}
      \fill[#1] (\cx,\cy) circle[radius=42.04];
  \end{scope}}
\newcommand{\focimark}[2][black]{%
  \setlength{\foci@mw}{#2}%
  \begin{tikzpicture}[x={\foci@mw/564.03}, y={-\foci@mw/564.03}]
    \useasboundingbox (0,0) rectangle (564.03,540.96);
    \focidrawmark{#1}{#2}{(0,0)}
  \end{tikzpicture}}

\fancypagestyle{focititlepage}{\fancyhf{}\renewcommand{\headrulewidth}{0pt}}

\newcommand{\foci@headerline}{%
  \noindent
  \begin{minipage}[c]{0.5\linewidth}
    \includegraphics[height=0.58cm]{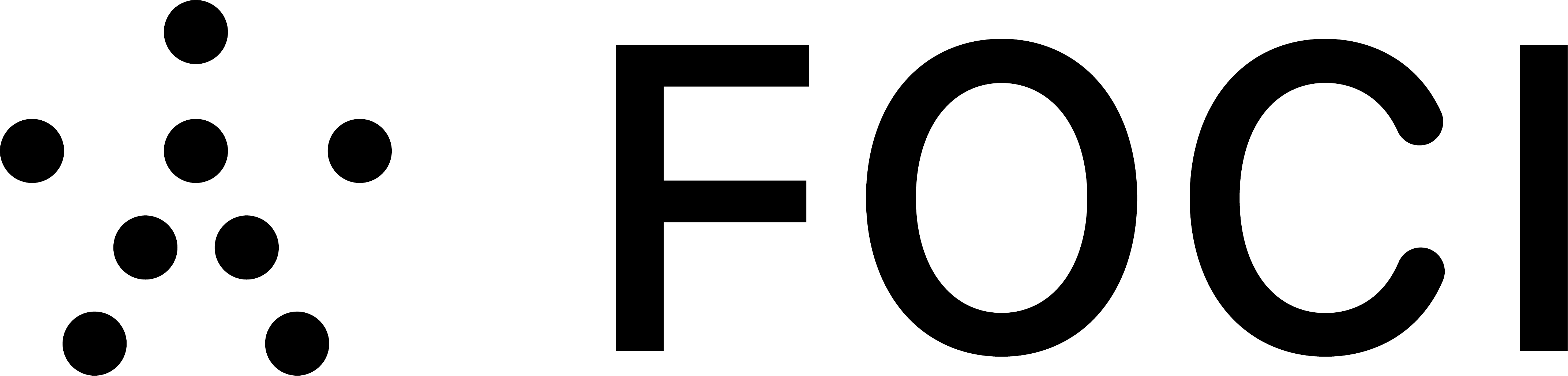}
  \end{minipage}%
  \begin{minipage}[c]{0.5\linewidth}
    \raggedleft\small\itshape\color{black!55}\foci@date
  \end{minipage}\par}

\newcommand{\foci@contactline}{%
  {\footnotesize\color{black!65}%
   \textbf{Email:}~\texttt{\foci@email}\qquad
   \textbf{Project:}~\url{\foci@project}\par}}

\newcommand{\foci@hairline}{\noindent\textcolor{black}{\rule{\linewidth}{0.5pt}}\par}

\newcommand{\focicoverCard}{%
  \thispagestyle{focititlepage}%
  \foci@headerline
  \vskip -0.5em
  \foci@hairline
  \vskip 0.9em
  \begin{tcolorbox}[enhanced, frame hidden, boxrule=0pt, arc=12pt,
      colback=black!4, width=\linewidth,
      left=18pt, right=18pt, top=18pt, bottom=16pt,
      grow to left by=2pt, grow to right by=2pt]
    \noindent
    \begin{minipage}[t]{\linewidth}
      {\setstretch{1.15}\LARGE\bfseries\raggedright Sign-Constrained Intervention Effects for Domain-Generalizable ICU World Models
\unskip\par}
      \vskip 0.55em
      {\normalsize\fociauthorlist\par}
      \vskip 0.2em
      {\small\color{black!70}\foci@affil\par}
    \end{minipage}\par
    \vskip 0.9em
    {\small\setstretch{1.08}Predicting how a patient's vital signs will respond to an intervention is a central question in intensive care. World models can do so by learning dynamics as a function of prior actions, so that a trajectory can be unrolled under any chosen treatment course. However, such models tend to be brittle for dosing practices outside of those seen in training data. Clinicians choose drug dosages based on the patient's state, so the association a model learns between dose and outcome runs opposite to the drug's \textit{effect}, which does not generalize well to out-of-distribution (OOD) dosing practices. However, pharmacological information already provides the effects on vital signs of different drugs, which is invariant to dosing practice distribution shifts. To resolve the OOD sensitivity issues of current ICU world models, we introduce \ourmethod, a world model that isolates an intervention's effect by evaluating a forward pass twice: once under the recorded doses, and once with those doses set to zero. The difference is then projected onto the set of pharmacologically admissible \textit{effect directions} stated by a corpus of 70 (drug, channel, sign) rules over 28 drugs, leaving the magnitude, which no source states, to be learned from data. Across five distribution-shift parameters over cohorts from three intensive-care databases, and against forecasters, counterfactual treatment-effect models, invariance objectives and two weaker placements of the same knowledge, \ourmethod\ attains the lowest OOD RMSE in every setting, improving on the strongest external baseline by 8--12\%, while matching the in-distribution error of an unconstrained backbone. Reversing the pharmacological effect directions degrades accuracy below the unconstrained backbone, and directions estimated from the training records recover almost no performance gain, indicating that the improvement derives from the content of the pharmacological knowledge rather than from the presence of a constraint.
The construction is not tied to intensive care and may extend to other settings where the sign of an effect is known in advance and its magnitude is confounded by how the treatment was assigned.
\unskip\par}
    \vskip 0.9em
    \foci@contactline
  \end{tcolorbox}
  \vskip 0.4em}

\newcommand{\focicover}{\focicoverCard}

\makeatother

\newcommand{\ourmethod}{\textsc{PhysioWorld}\xspace}
\definecolor{tableblue}{HTML}{315A7D}
\definecolor{tableblueLight}{HTML}{EAF2F8}
\definecolor{tableblueVeryLight}{HTML}{F5F8FB}
\definecolor{tablegray}{HTML}{66727C}
\definecolor{tableredLight}{HTML}{FBF0EF}
\newcommand{\ourval}[1]{\textcolor{tableblue}{#1}}
\newcolumntype{Y}{>{\centering\arraybackslash}X}
\begin{document}

\fociauthor{Zhen Xu}{1}
\fociauthor{Nicholas Konz}{1}
\fociauthor{Zhen Tan}{3}
\fociauthor{Zachary Plotkin}{2,\focicorr}
\fociauthor{Tianlong Chen}{1,2,\focicorr}

\fociaffil{\textsuperscript{1}University of North Carolina at Chapel Hill\quad
  \textsuperscript{2}Foci Labs\quad
  \textsuperscript{3}Stevens Institute of Technology}
\fociemail{tianlong@cs.unc.edu}
\focicorremail{zach@fo.ci, tianlong@fo.ci}
\fociproject{https://shahinxu.github.io/physioworld/}

\pagestyle{fancy}
\fancyhf{}
\renewcommand{\headrulewidth}{0.4pt}
\chead{\small\scshape Sign-Constrained Intervention Effects for ICU World Models
\unskip}
\cfoot{\small\thepage}

\focicover
\focifootnotes

\section{Introduction}

In an intensive care unit, clinicians must frequently choose interventions to perform on a patient, such as raising a vasopressor, starting a sedative, or changing a ventilator setting. Each decision rests on an implicit prediction of how the patient's vital signs will respond, and making that prediction explicit is a central computational problem in this setting \citep{hyland2020early, komorowski2018artificial}. The data is two streams on one clock: vital signs sampled by clinical routine as much as by physiology \citep{agniel2018biases, hinrichs2024short, harutyunyan2019multitask}, and sparse, asynchronous clinician actions that mix drug infusions, ventilator set-points and discrete switches. Clinical time-series models handle such data natively \citep{che2018recurrent, rubanova2019latent, nie2022time} but answer only ``what happens next'', whereas the question here is what happens under a treatment course one names, which requires interventions to be explicit inputs whose consequences propagate through the predicted physiology, the defining property of a \emph{world model} \citep{ha2018recurrent, lecun2022path}.

The difficulty is not in fitting such a model, but in applying it to units, eras, or dosing practices unlike those of the records it was trained on, which is where it is used and where every data-driven one fails. That failure is documented: on sepsis cohorts from MIMIC-IV and eICU, \citet{park2024consistent} trained 81 transformer dynamics models and found that replacing the recorded doses at test time with zeros or with a random permutation left prediction error essentially unchanged. What they carry is not a dose-response, and a model that holds none has nothing that a change of regime could preserve.

The cause is in the data rather than the architectures. A clinical record contains only the treatments that were given, and the dose is itself a reading of how sick the patient was when it was set \citep{gottesman2019guidelines, kotani2023updated}. A vasopressor appears in the record because the patient was deteriorating, so a model that reads the dose as a cause rather than as a symptom of severity learns that the drug makes the pressure fall \citep{walker1996confounding, hernan2010causal, ioannidis2001comparison}. That accounts for both symptoms: where dose and severity are collinear the fitted effect collapses toward zero, and where severity dominates it takes the wrong sign. Either way it surfaces out of distribution, since a dose-response learned this way survives only the regime it was fitted in.

The failure concentrates on one axis: every purely data-driven model we test degrades further out of distribution under a shift in treatment intensity than under any shift in the patient population (Section~\ref{sec:main}). A model whose dose-response is absent or reversed has nothing to carry across a change in dosing, while the rest of the physiology it learned survives intact. The \emph{direction} of a drug's effect, by contrast, is invariant across that shift: magnitudes vary with age, organ function and comorbidity, but a vasopressor raises blood pressure across patients and care units, and pharmacology states as much before any data is collected \citep{hall2016guyton}.

\ourmethod\ turns that asymmetry into an architecture. A standard forward pass does not expose the quantity a direction would constrain, since it outputs the next state as a whole, without separating the part caused by any single drug. \ourmethod\ recovers that part by evaluating the same predictor twice with identical weights and history, once with the recorded doses and once with those doses set to zero, and a closed-form, parameter-free projection then maps the difference onto the cone of pharmacologically admissible directions. The operator therefore constrains signs alone and leaves magnitudes to training (Section~\ref{sec:method}). We evaluate \ourmethod\ on five distribution-shift axes across three ICU databases, where it exceeds every baseline out of distribution at the in-distribution error of its backbone, and the gain holds under free-running rollout (Section~\ref{sec:experiments}).

This paper makes three contributions.
\begin{itemize}[nosep,leftmargin=*]
\item We build a shift benchmark on five cohorts from three intensive-care databases, spanning patient population and treatment intensity, and show that clinical dynamics models degrade further under a shift in how much drug is given than under any shift in who the patient is.
\item We close that gap with \ourmethod, a clinical world model that isolates its estimate of an intervention's effect with a second forward pass and projects that estimate onto the stated pharmacological directions, a parameter-free, provably non-expansive operation that reduces out-of-distribution RMSE by up to 12\% over the strongest external baseline.
\item We ablate the injected corpus out of distribution. Inverting the stated directions is worse than injecting none, and data-mined directions recover little of the benefit, so the generalization gain comes from the content of the pharmacology rather than from the act of constraining.
\end{itemize}

\section{Related Work} \label{sec:related}

\textbf{World models.} Action-conditioned latent dynamics were first developed in games and robotics \citep{ha2018recurrent, hafner2019dream, hafner2025mastering}, and joint-embedding predictive architectures (JEPA) later showed that predicting in representation space, without reconstructing the input, is sufficient for the same purpose \citep{lecun2022path, assran2023self, bardes2024revisiting, assran2025v, balestriero2025lejepa, bagatella2026td, chen2026vl}, a recipe our backbone follows. Several world models have recently been proposed for medicine \citep{liu2026medical, yang2026clin, yang2025medical, wang2026chronomedicalworld, mu2026ehrworld, liu2026drift}, but none of the six applies to the present task: four take no intervention input at all, one models an intervention as producing an empty outcome with probability one, and one encodes vasopressor exposure as a binary flag, so a treatment course cannot be specified. None of the six builds pharmacological structure into the architecture, and none is evaluated under distribution shift.

\textbf{Clinical prediction, interventions and shift.} Models for bedside data address irregular sampling directly \citep{che2018recurrent, de2019gru, rubanova2019latent, yalavarthi2024grafiti}, while windowed and pretrained forecasters remain strong on aggregate error \citep{zeng2023transformers, nie2022time, liu2024itransformer, wang2024timexer, ansari2025chronos, li2026mira, gu2024time, pucher2026can, meyer2025rethinking}, and in all of them interventions enter as covariates. A second line estimates counterfactual outcomes under treatment that changes over time, adjusting for time-dependent confounding with balanced representations \citep{bica2020estimating} or attention over treatment history \citep{melnychuk2022causal}; both are baselines here, and both face the confounding described above \citep{walker1996confounding, hernan2010causal, gottesman2019guidelines}, which is why only knowledge that does not require estimating response magnitudes from the record can be injected safely. Robustness to shift is otherwise sought through invariant features \citep{arjovsky2019invariant, sun2016deep, sagawa2019distributionally}, and clinical shift benchmarks stratify patients by demographics, care unit or era \citep{avati2021beds, koh2021wilds, yang2023change}, all of which describe who the patient is.

\textbf{Knowledge-constrained learning.} A common approach relaxes logic or physics constraints into a differentiable penalty \citep{xu2018semantic, raissi2019physics, badreddine2022logic, huang2026knowledge}, which leaves the model free to violate the rules it was penalized with. A second approach enforces them exactly by projecting the output onto a feasible set \citep{amos2017optnet, donti2021dc3, dalal2018safe}. A third builds the constraint into the function class, as in monotonic networks \citep{sill1997monotonic, you2017deep, liu2020certified} and physics-informed architectures \citep{greydanus2019hamiltonian, cranmer2020lagrangian}, and pharmacological priors have entered clinical state-space models in a similar way \citep{hussain2021neural, qian2021integrating}. All of these constrain a predicted \emph{value}, or its derivative with respect to a designated input, and none reaches the quantity at issue: an intervention effect can contradict the direction of a drug's action while every value it produces stays feasible, because monotonicity in a raw input places no constraint on a dose effect that has already propagated through the latent representation, as our Clamp baseline shows. We therefore place the constraint on $\Delta z$ itself, inside the recurrent transition.

\section{Problem Setup} \label{sec:problem}

\begin{wrapfigure}{r}{0.34\textwidth}
    \centering
\includegraphics[width=\linewidth]{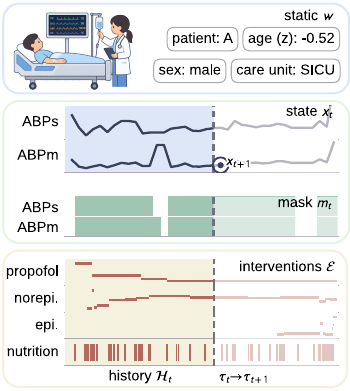}
    \caption{A real ICU stay illustrating the task.}
    \label{fig:problem_setup}
    \vspace{6pt}
\end{wrapfigure}
\textbf{Problem formalization.} A patient stay is a triple of streams on a common clock (Figure~\ref{fig:problem_setup}). The \emph{state stream} is a sequence $\{(\tau_t, x_t, m_t)\}_{t=1}^{T}$ of irregularly spaced times $\tau_1 < \cdots < \tau_T$, last-observed vital-sign vectors $x_t \in \mathbb{R}^{C}$, and freshness masks $m_t \in \{0,1\}^{C}$, where $m_{t,c}=1$ if and only if channel $c$ is newly measured at $\tau_t$. Between measurements $x$ carries the last value forward, so the mask, not the value, encodes what is known. The \emph{intervention stream} is a set of typed events $\mathcal{E} = \{(s_i, k_i, v_i)\}$: administration time $s_i$ in continuous hours, intervention channel $k_i \in \{1,\ldots,A\}$ (drug infusions and boluses, ventilator set-points, discrete switches), and dose or setting $v_i$. The third stream is the static covariate vector $\xi$ (e.g.\ age, sex, care unit), fixed for the stay. The task is one-step state prediction under a specified intervention course: given the history $\mathcal{H}_t$ of all three streams up to $\tau_t$ and the interventions in $(\tau_t, \tau_{t+1}]$, predict $x_{t+1}$ at the next observed time,
\begin{equation}
\hat{x}_{t+1} = f\big(\mathcal{H}_t,\ \mathcal{E} \cap (\tau_t, \tau_{t+1}],\ \tau_{t+1} - \tau_t\big).
\label{eq:task}
\end{equation}
Multi-step prediction is a free-running rollout in which predictions are fed back as states while the future intervention course is held to a specified plan, and because that course is an explicit argument of $f$, separate from the history, substituting an alternative course poses a ``what-if'' query. The gap between consecutive states, $\tau_{t+1}-\tau_t$, is also an input, since the model must predict across whatever horizon the chart imposes.

\textbf{Out-of-distribution evaluation.} The quantity this paper reports is not the error of $f$ on stays like the ones it was fitted on. A partition cuts the cohort into slices $\mathcal{D}_1,\dots,\mathcal{D}_K$ along one axis, such as the intensity of a drug class or the age of the patient; a model fitted on $\mathcal{D}_i$ is then scored on the held-out stays of every $\mathcal{D}_j$. The diagonal $i=j$ is the in-distribution (ID) error, the off-diagonal entries $i \neq j$ are the out-of-distribution error, and it is the latter that every result below refers to.

\textbf{Stream representation.} \label{sec:streams} The intervention stream is a sparse, asynchronous event set, while the network needs a value at every prediction time, so we bridge the two with an \emph{exposure field} $u(s)$, whose entry $u_k(s)$ states how strongly intervention channel $k$ is acting on the patient at continuous time $s$. A first-order elimination kernel, $u_k(s) = \sum_i v_i e^{-(s-s_i)/\theta_k}$, is the usual choice \citep{hall2016guyton}, but it writes a standard pharmacological model into a representation that this paper otherwise keeps free of pharmacological assumptions, so that any directional structure the model later exhibits can be attributed to the injected rule corpus rather than to this encoding. We therefore instantiate $u$ as a bucketed dose history, summing the dose of each channel within five fixed elapsed-time windows whose boundaries encode a coarse timescale, but the same one for every channel, so that no drug receives a fitted constant. Infusions, boluses and ventilator set-points enter this representation differently, and the state encoder additionally receives a per-channel freshness feature computed from the time since that channel was last measured; Appendix~\ref{app:streams} gives the window boundaries, the treatment of each intervention type, and the freshness definition.

\section{PhysioWorld} \label{sec:method}

We now introduce \ourmethod, which separates what a forward pass predicts from what pharmacology is able to constrain. The separation is what buys the out-of-distribution behavior: the rules the operator enforces are fixed before any record is read, so they are the same under every dose regime, while everything estimated from data is left where the shift can reach it. Three objects are fixed before training begins: the rule corpus, the sign cone it defines at each step, and the projection onto that cone, which has no parameters. Two are learned: the backbone that produces the intervention effect, and the calibration map $G$ in whose rows that effect is read, whose directions an anchor term pins to the channels so that the signs cannot move with the dose regime.

\begin{figure}[t]
\centering
\includegraphics[width=\textwidth]{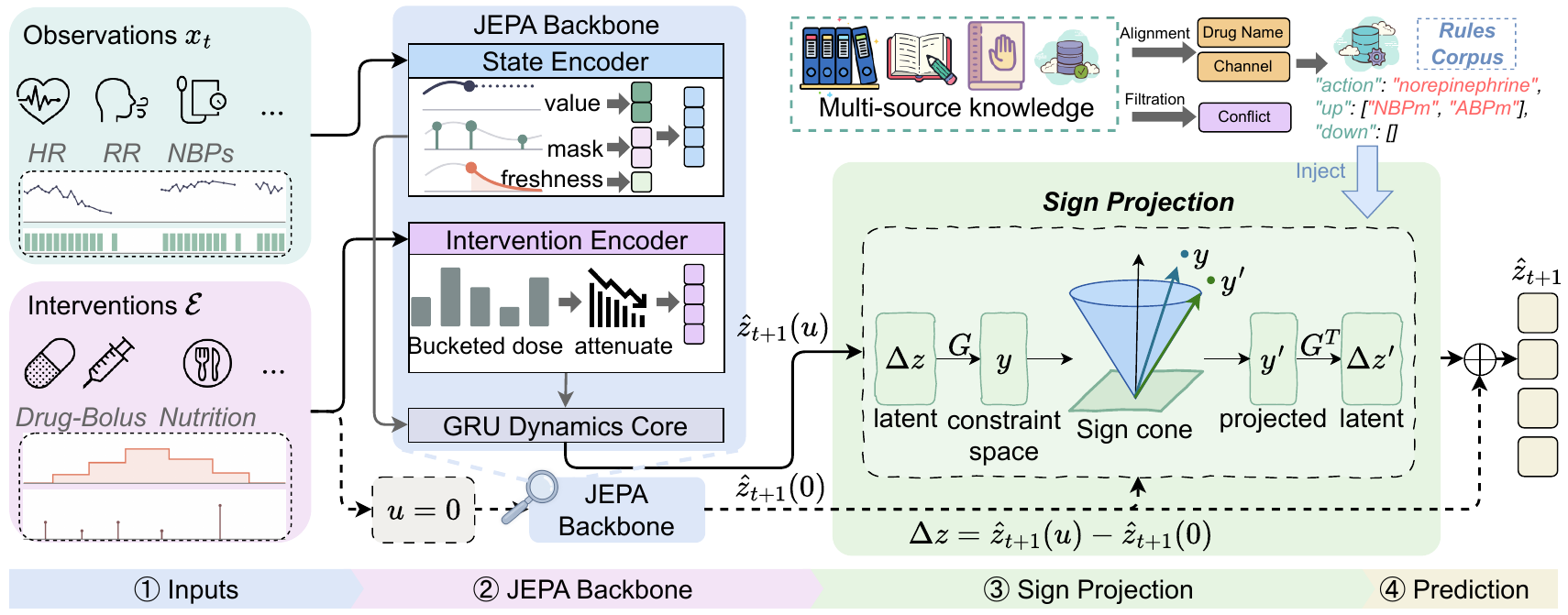}
\caption{\ourmethod's forward pass in four stages. (1) States $x_t$ and interventions $\mathcal{E}$ are the inputs. (2) The state and intervention encoders feed a GRU core that runs once under the recorded doses $u$ and once with $u=0$, giving $\hat{z}_{t+1}(u)$ and $\hat{z}_{t+1}(0)$. (3) Their difference $\Delta z$ is mapped by $G$ into channel units $y$ and projected onto the sign cone to give $\Pi(y)$. (4) The correction $G^\top(\Pi(y) - y)$ is added back into the latent to give $\hat{z}_{t+1}$.}
\label{fig:world_model}
\end{figure}

\textbf{Rule corpus.} \label{sec:rules} We inject a corpus of \emph{intervention rules} that state, for each of 28 drugs, the sign/direction of the drug's effect on each channel it touches, 70 (intervention, channel) pairs in all. We keep the sign and omit the magnitude, because dose is titrated against the quantity it changes: a vasopressor is increased as pressure falls and withdrawn as it recovers, so the observed association between dose and effect runs opposite to the drug's action, and the magnitude of the effect cannot be reliably recovered from observational data \citep{walker1996confounding}. The direction, by contrast, is a statement about pharmacological mechanism and is unaffected by this confounding variable. Constraints on the state itself, such as per-channel physiological ranges and the ordering $\text{systolic} \ge \text{mean} \ge \text{diastolic}$, bound which values are possible, but we do not use them, because the failure we target lies in how a prediction \emph{reacts} to an intervention, and that reaction can be wrong while every value stays inside the feasible set. Appendix~\ref{app:rules} lists the corpus and its coverage by channel.

\textbf{A second forward pass.} \label{sec:decomposition} A rule about a drug constrains what that drug does, but a standard forward pass does not isolate that quantity: it reports the next latent with the intervention's contribution and the patient's own course already summed. We recover that contribution by holding the weights and the recorded observations fixed and advancing the same predictor twice along the stay (Figure~\ref{fig:world_model}), once under the recorded dose entries $u$ and once with the dose entry of every covered drug zeroed at every step,
\begin{equation}
\Delta z_{t+1} \;=\; \hat{z}_{t+1}(u) \;-\; \hat{z}_{t+1}(0).
\label{eq:delta}
\end{equation}
Both passes read the same states, so the encodings $z_{1:t}$ are identical and none of the difference originates in the observed course; only the dose entries of the intervention embedding $a_s$, $s \le t$, of Section~\ref{sec:backbone} differ, and the non-dose intervention channels are untouched. The recurrent states separate at the first administered dose and stay apart, since the intervention encoding enters the recurrence, so $\Delta z_{t+1}$ is the contribution of the dose history as a whole rather than of the dose delivered in $(\tau_t, \tau_{t+1}]$ alone. That contribution is dominated by the recent past, a trailing window of $w = 3$ steps recovering it to within $1\%$ (Appendix~\ref{app:window}), which is what licenses the cone of Section~\ref{sec:operator} to be built from the drugs active at $t$. The two passes advance as parallel recurrent states in one sweep, at the cost of one extra forward pass. Appendix~\ref{app:design} discusses withdrawing the non-dose channels as well, and severing the intervention stream from the recurrent core so that $u=0$ would return the untreated course exactly.

\textbf{The sign projection.} \label{sec:operator} Let $\mathcal{K}_t$ be the corpus-covered (drug, channel) pairs whose drug is active at step $t$. Where every pair in $\mathcal{K}_t$ touching a channel $c$ agrees on the sign of its effect, write $\sigma_c$ for that common sign; where they disagree, as when a vasopressor and a sedative run together, no sign is implied and none is imposed. In practice, conflict is infrequent: pooled over all channels, it affects $2.3\%$ of scored entries and $4.1\%$ of those under treatment (Appendix~\ref{app:rules}, Table~\ref{tab:channel_rules}). Write $y = G\,\Delta z_{t+1}$ for the effect in channel units. $G$ is a $C \times d$ calibration map with orthonormal rows ($C \le d$) carrying the $d$-dimensional latents of Section~\ref{sec:backbone} into the $C$ channels, trained by an anchor term alone and separate from the readout head: $\mathcal{L}_{\mathrm{anchor}} = |\mathcal{M}|^{-1}\sum_{(t,c) \in \mathcal{M}} ((G z_t)_c - \tilde{x}_{t,c})^2$, over the charted entries $\mathcal{M}$ only, with $\tilde{x}$ the standardized channel values. Row $c$ is thus fit to read channel $c$, which fixes its sign rather than leaving it free: the flipped row predicts $-\tilde{x}_c$ and the term rejects it. Under the default training, no scored row of $G$ reverses across seeds and initializations (Appendix~\ref{app:gcal}). Because $G$ is linear, an anchor fit on levels transfers to differences without further assumptions, $G(z_1 - z_2) = G z_1 - G z_2$, so $(G \Delta z)_c$ reads the same channel that $(G z)_c$ does. The operator then projects onto the resulting cone and returns the correction to the latent,
\begin{equation}
\Pi(y)_c \;=\;
\begin{cases}
\sigma_c \max(\sigma_c y_c,\, 0), & \sigma_c \text{ defined},\\[2pt]
y_c, & \text{otherwise},
\end{cases}
\qquad
\hat{z}_{t+1} \;=\; \hat{z}_{t+1}(0) + \Delta z_{t+1} + G^{\top}\!\left(\Pi(y) - y\right).
\label{eq:operator}
\end{equation}
Write $\Delta z'_{t+1} = \Delta z_{t+1} + G^{\top}(\Pi(y) - y)$ for the corrected effect, so that $\hat{z}_{t+1} = \hat{z}_{t+1}(0) + \Delta z'_{t+1}$. The corrected $\hat{z}_{t+1}$ is the state advanced into step $t{+}1$, so the correction enters the recurrence rather than acting only at the readout; this is what separates the operator from filtering the decoded output. Algorithm~\ref{alg:step} states the per-step procedure in full.
\begin{proposition}[Properties of the sign projection]
\label{prop:nonexpansive}
Let $\mathcal{S} = \{\, y \in \mathbb{R}^{C} : \sigma_c y_c \ge 0 \text{ for every } c \text{ on which } \sigma_c \text{ is defined} \,\}$, and let $\Pi$ be as in \eqref{eq:operator}. Then $\mathcal{S}$ is a nonempty closed convex cone, $\Pi$ is the Euclidean projection onto it, and for every $y \in \mathbb{R}^{C}$:
\begin{enumerate}[nosep,label=(\roman*),leftmargin=*]
\item $\Pi(y) \in \mathcal{S}$;
\item $\Pi(y) = y$ whenever $y \in \mathcal{S}$;
\item $\|\Pi(y) - \Pi(y')\| \le \|y - y'\|$ for every $y'$;
\item for every rule-consistent $y^{\star} \in \mathcal{S}$,
\[
\|y - y^{\star}\|^{2} - \|\Pi(y) - y^{\star}\|^{2}
\;\ge\; \|y - \Pi(y)\|^{2}.
\]
\end{enumerate}
\end{proposition}
All four properties are elementary, and the proof is in Appendix~\ref{app:proj}.

\begin{remark}[Scale]
\label{rem:scale}
$\mathcal{S}$ is coordinatewise, so rescaling row $c$ of $G$ by any $\rho_c > 0$ leaves unchanged which components $\Pi$ clips and no per-cohort calibration of scale is needed. Orthonormal rows fix the size of the write-back as well as its direction, and $G G^{\top} = I_C$ makes $G\,\Delta z'_{t+1} = \Pi(y)$ exact, so the corrected latent lands in the cone rather than near it.
\end{remark}

Since $G$ trains with the projection active, a $G$ on which the constraint never binds would reduce the operator to the identity at no cost in accuracy; the projection therefore reads a detached $G$, which removes rotating $G$ as a descent direction. The remaining route, collapsing $\Delta z$, is measured rather than forbidden: the projection moves the response on $54$ to $61\%$ of the scored entries on which a sign is defined, and Table~\ref{tab:gcal} varies the initialization, the anchor weight, the orthonormality of the rows and the fitting protocol. The guarantee is also limited: it fixes signs but not magnitudes, it imposes nothing where active drugs disagree or where the corpus is silent, and it holds for $G\,\Delta z$ and not for the decoded prediction, which is why Section~\ref{sec:dose} measures agreement on the decoded output. Convexity is a further reason the rule corpus is stated as signs, since the admissible set of a conditional rule is a union of half-spaces, on which the nearest point need not be unique and non-expansiveness is lost.

\textbf{JEPA backbone and training.} \label{sec:backbone} The backbone is a joint-embedding predictive architecture over the represented streams, in which an encoder maps each step's state $[x_t, m_t, \phi_t]$ to a latent $z_t$ and a recurrent core consumes $[z_t, a_t, \delta_t]$. Here $a_t$ embeds the bucketed intervention history introduced in the stream representation, and $\delta_t$ contains the gap behind step $t$ and the gap ahead to the target, each on a log and a bounded scale, together with a continuous-time encoding of $\tau_t$, and a predictor then maps the recurrent latent and $a_t$ to $\hat{z}_{t+1}$. Training follows LeJEPA \citep{balestriero2025lejepa}: the target is the same online encoder's encoding of the true next state, with no EMA teacher and no stop-gradient, and collapse is prevented by a distributional regularizer that pushes the embedding distribution toward an isotropic Gaussian, giving $\mathcal{L}_1 = (1-\lambda)\,\|\hat{z}_{t+1} - z_{t+1}\|^2 + \lambda\, \mathcal{L}_{\mathrm{SIGReg}}$. Training has two stages: the first trains the encoder, predictor and intervention pathway without any physical-space loss, and the second freezes the backbone and trains only a readout, so that no physical-space gradient reaches the representation and decoded accuracy can be attributed to what the representation encodes \citep{bardes2024revisiting}. The full objective adds the anchor term of Section~\ref{sec:operator}, $\mathcal{L} = \mathcal{L}_1 + \lambda_{\mathrm{anchor}} \mathcal{L}_{\mathrm{anchor}}$ with $\lambda_{\mathrm{anchor}} = 1$, and carries it through both stages: $G$ is the map the cone is defined on, so it is anchored while the representation forms and again while the readout trains. Inside the operator $G$ enters under stop-gradient, so $\mathcal{L}_1$ trains the backbone through the corrected latent but cannot rotate $G$ to make the projection vacuous; $G$ is updated by $\mathcal{L}_{\mathrm{anchor}}$ alone. The readout decodes a delta on the carry-forward anchor, $\hat{x}_{t+1} = x_t + \mathrm{MLP}(\hat{z}_{t+1})$, which is the same output parameterization every baseline uses.

\section{Experiments} \label{sec:experiments}

Section~\ref{sec:main} asks which kind of distribution shift damages these models and whether the pharmacology repairs it, Section~\ref{sec:rollout} whether the damage compounds when a model runs on its own predictions, Section~\ref{sec:dose} whether the dose is used in the direction pharmacology states, and Section~\ref{sec:corpus} where the gain comes from, by varying the injected corpus and by removing one component at a time.

\subsection{Setup} \label{sec:setup}

\textbf{Datasets.} We evaluate five cohorts from three databases, selected on the primary discharge diagnosis so that every stay belongs to exactly one. \textbf{MIMIC-IV} \citep{johnson2023mimic} covers 2008 to 2019 under MetaVision charting and contributes circulatory (20{,}128 stays) and injury (7{,}746) cohorts; \textbf{MIMIC-III} \citep{johnson2016mimic}, restricted to its CareVue era of 2001 to 2008, contributes 10{,}020 and 4{,}480 stays for the same two chapters and is independent of MIMIC-IV in charting system, period and patients; \textbf{VitalDB} \citep{lee2022vitaldb} adds 5{,}391 surgical cases from a tertiary center in Seoul, gridded at five minutes rather than one hour. The same corpus is applied to all five without modification, which tests whether a stated direction can be \emph{reused as written}.

\textbf{Task.} The task is the one-step prediction of Section~\ref{sec:problem} for the nine vital-sign channels defined in Appendix~\ref{app:rules}. The intervention stream covers drugs, blood products and ventilator set-points, with 35 channels on MIMIC-IV, 30 on MIMIC-III and 12 on VitalDB, and predictions are scored only on freshly measured channels whose carry-forward anchor was itself measured.

\textbf{Shift axes.} Axes are built by a \emph{split-then-slice} protocol: following \citet{gulrajani2020search}, the cohort is partitioned once by a global 80/20 subject-wise split, each slice is the intersection of a stratum with that split, and training sets are downsampled to the smallest slice so that data volume does not confound the comparison. \emph{Population axes} stratify patients by characteristics fixed on admission; \emph{intervention-intensity axes} band exposed stays by within-class dose rank, taking the bottom and top quartiles of burden, which assumes no equipotency across drugs. There are five: age, disease subtype and era, cut into 5, 4 and 4 slices, and vasopressor and sedative. VitalDB has too few exposed stays for the vasopressor axis, so cross-database comparison uses the sedative axis alone (Appendix~\ref{app:holdouts}). Test slices differ in intrinsic difficulty \citep{koh2021wilds}, so carry-forward persistence, which scores zero skill by definition, is the reference for the ID column.

\textbf{Baselines.} All methods share the same grid and carry-forward anchor \citep{zeng2023transformers, qiu2024tfb}. The forecasters are PatchTST \citep{nie2022time}, GraFITi \citep{yalavarthi2024grafiti} for irregular sampling, and the pretrained Chronos-2 \citep{ansari2025chronos}, all three receiving the planned intervention course as a known-future covariate. CRN \citep{bica2020estimating} and the Causal Transformer \citep{melnychuk2022causal} forecast under counterfactual time-varying interventions, the setting closest to ours, while IPTW reweights training stays by inverse treatment probability and IRM \citep{arjovsky2019invariant} and CORAL \citep{sun2016deep} seek features invariant across training slices. JEPA, Soft constraint and Clamp share our backbone and differ only in where the constraint acts: \emph{nowhere, in the loss, or on the output}.

\textbf{Metrics.} We report RMSE in physical units, in and out of distribution, as an unweighted macro mean over channels \citep{zhang2024irregular}. A \emph{skill score} against persistence, $1-\mathrm{RMSE}/\mathrm{RMSE}_{\text{pers}}$, is computed within each slice pair and aggregated geometrically \citep{hamill2006measuring, davydenko2013measuring}; it is the primary metric across databases, whose sampling rates make absolute RMSE incomparable. We also report \emph{dispersion}, the ratio of the across-patient standard deviation of a model's predictions to that of the truth, where values below one indicate collapse toward the population mean. All numbers are mean and standard deviation over five seeds. Differences against Clamp, the strongest method sharing our backbone, are tested with a patient-level paired bootstrap and a one-sided Wilcoxon signed-rank test over train/test slice pairs (Appendix~\ref{app:stats}).

\subsection{Shift axes: population against intervention intensity} \label{sec:main}

\begin{table}[t]
\caption{In-distribution and out-of-distribution RMSE ($\downarrow$) on MIMIC-IV, physical units, mean $\pm$ std over five seeds. The ID column is a model trained on the full training split and tested on the full test split. Each OOD entry averages the off-diagonal of that axis's train-slice $\times$ test-slice matrix. \textbf{Bold}: best per shift axis. \underline{Underline}: best external baseline. The ID column is unmarked by design.}
\label{tab:main}
\begin{center}
\scriptsize
\setlength{\tabcolsep}{2pt}
\renewcommand{\arraystretch}{0.95}
\begin{tabularx}{\textwidth}{@{}l *{6}{Y}@{}}
\toprule
& 
& \multicolumn{3}{c}{\textbf{Population shift}}
& \multicolumn{2}{c}{\textbf{Intervention-intensity shift}} \\
\cmidrule(lr){3-5}
\cmidrule(l){6-7}
Method
& In-dist
& Age
& Disease
& Era
& Vasopressor
& Sedative \\
\midrule
\multicolumn{7}{@{}l}{
    \textcolor{tablegray}{\emph{Reference and data-driven forecasters}}
} \\
Persistence
& 10.57
& 10.86
& 11.07
& 10.86
& 10.83
& 10.91 \\
PatchTST
& $8.93\pm0.05$
& $9.68\pm0.07$
& $9.18\pm0.07$
& $8.92\pm0.06$
& $12.15\pm0.08$
& $11.60\pm0.08$ \\
GraFITi
& $8.95\pm0.05$
& $9.62\pm0.07$
& $9.21\pm0.07$
& $8.95\pm0.06$
& $11.85\pm0.08$
& $11.42\pm0.07$ \\
Chronos-2
& $8.88\pm0.05$
& \underline{$9.55\pm0.06$}
& \underline{$9.12\pm0.06$}
& \underline{$8.86\pm0.06$}
& $11.70\pm0.08$
& $11.35\pm0.07$ \\
CRN
& $9.34\pm0.06$
& $9.85\pm0.07$
& $9.45\pm0.07$
& $9.20\pm0.06$
& $10.95\pm0.08$
& $10.68\pm0.07$ \\
Causal Transformer
& $9.08\pm0.05$
& $9.72\pm0.06$
& $9.30\pm0.06$
& $9.05\pm0.06$
& \underline{$10.45\pm0.07$}
& \underline{$10.28\pm0.07$} \\
\midrule
\multicolumn{7}{@{}l}{
    \textcolor{tablegray}{\emph{Robustness by reweighting or invariance}}
} \\
IPTW
& $8.12\pm0.06$
& $9.65\pm0.07$
& $10.55\pm0.08$
& $9.40\pm0.07$
& $10.55\pm0.07$
& $10.38\pm0.07$ \\
IRM
& $9.31\pm0.06$
& $10.75\pm0.09$
& $11.12\pm0.09$
& $10.65\pm0.08$
& $12.35\pm0.10$
& $12.10\pm0.09$ \\
CORAL
& $10.21\pm0.07$
& $10.35\pm0.08$
& $10.48\pm0.08$
& $10.15\pm0.07$
& $11.20\pm0.08$
& $10.95\pm0.08$ \\
\midrule
\multicolumn{7}{@{}l}{
    \textcolor{tablegray}{\emph{World models}}
} \\
JEPA
& $7.95\pm0.04$
& $9.34\pm0.08$
& $10.68\pm0.08$
& $9.46\pm0.07$
& $11.15\pm0.08$
& $11.60\pm0.08$ \\
Soft constraint
& $7.90\pm0.04$
& $9.25\pm0.06$
& $8.85\pm0.06$
& $8.62\pm0.05$
& $11.10\pm0.07$
& $10.75\pm0.07$ \\
Clamp
& $7.92\pm0.04$
& $8.55\pm0.05$
& $8.40\pm0.05$
& $8.20\pm0.04$
& $9.65\pm0.06$
& $9.38\pm0.06$ \\
\rowcolor{tableblueLight}
\textbf{\textcolor{tableblue}{\ourmethod}}
& $7.95\pm0.04$
& $\mathbf{8.41\pm0.05}$
& $\mathbf{8.26\pm0.05}$
& $\mathbf{8.15\pm0.04}$
& $\mathbf{9.42\pm0.06}$
& $\mathbf{9.21\pm0.06}$ \\
\bottomrule
\end{tabularx}
\end{center}
\end{table}

\begin{figure}[t]
\centering
\includegraphics[width=\textwidth]{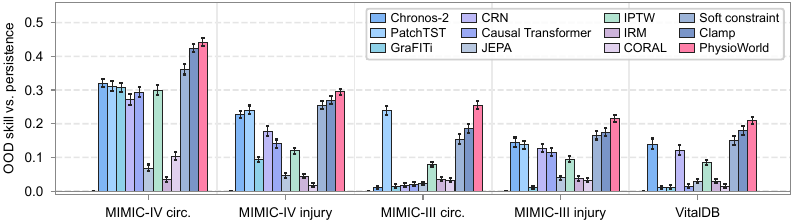}
\caption{Cross-database generalization under sedative intervention-intensity shift. Skill score is the fractional RMSE reduction over carry-forward persistence; higher is better, $0$ means no gain.}
\label{fig:cross_dataset}
\end{figure}

\textbf{Population shift.} Across \emph{Age}, \emph{Disease} and \emph{Era}, most models degrade moderately, which is consistent with demographic shift acting mainly on the covariate distribution $P(X)$ while the physiological transition dynamics stay largely invariant. Model capacity appears sufficient here, and the margin of \ourmethod\ over Clamp, the strongest baseline, is \emph{consistent but modest}.

\textbf{Intervention-intensity shift.} Under \emph{Vasopressor} and \emph{Sedative} shift the forecasters separate into two groups: \textbf{pure sequence models fall \emph{below} persistence} (PatchTST $12.15$, persistence $10.83$), while methods that model the intervention explicitly remain competitive with it (Causal Transformer, $10.45$). As explained in the Introduction, clinical dosing is confounded by indication, so a model fit to observational association reads dose as a marker of severity, and that the two groups diverge under the same shift suggests that \emph{where} the intervention enters the model matters more than model capacity. Domain-invariant objectives do not close the gap and fail in opposite ways (Section~\ref{sec:rollout}).

\textbf{How the knowledge is injected.} A soft monotonicity penalty appears to be dominated by the reconstruction gradients and does not bind out of distribution ($11.10$), and clamping is stronger ($9.65$) but acts \emph{only on the emitted prediction}, leaving unchanged the latent representation that produced the violation. \ourmethod, which constrains the signed dose--effect pathway \emph{at every recurrent transition}, is best on all five axes, with in-distribution error equal to its backbone's and a smallest margin of $0.05$ over Clamp on \emph{Era}. Figure~\ref{fig:cross_dataset} extends the comparison to five cohorts under sedative shift, where the gain is small but present on every cohort: every data-driven forecaster drops to within $0.02$ of persistence on at least one cohort, and which cohort it is differs by method. \ourmethod\ stays at or above $\ourval{0.21}$ on all five cohorts and exceeds the best baseline on each; its margin over Clamp is positive everywhere but within about two seed standard deviations on four of the five, and the bootstrap interval of the difference excludes zero on four of the five axes (Appendix~\ref{app:stats}). We therefore claim \emph{a consistent ordering, not a win on any individual cohort}.

\subsection{Horizon analysis under free-running rollout} \label{sec:rollout}

\begin{figure*}[t]
    \centering
    \begin{minipage}[b]{0.30\textwidth}
        \centering
        \scriptsize
        \setlength{\tabcolsep}{1.2pt}
        \renewcommand{\arraystretch}{0.95}
        \begin{tabular}{@{}lcc@{}}
        \toprule
        Method & Slope $\downarrow$ & Disp. ($\to 1$) \\
        \midrule
        Persistence        & $+0.084$ & $0.985$ \\
        PatchTST           & $+0.127$ & $0.631$ \\
        GraFITi            & $+0.116$ & $0.700$ \\
        Chronos-2          & $+0.124$ & $0.645$ \\
        CRN                & $+0.108$ & $0.758$ \\
        Causal Transformer & $+0.115$ & $0.663$ \\
        \addlinespace[1.5pt]
        IPTW               & $+0.296$ & $1.825$ \\
        IRM                & $+0.472$ & $1.757$ \\
        CORAL              & $+0.102$ & $0.415$ \\
        \addlinespace[1.5pt]
        JEPA               & $+0.150$ & $0.748$ \\
        Soft constraint    & $+0.141$ & $0.765$ \\
        Clamp              & $+0.116$ & $0.792$ \\
        \rowcolor{tableblueLight}
        \textbf{\textcolor{tableblue}{\ourmethod}}
                           & \textbf{+0.090} & \textbf{0.842} \\
        \bottomrule
        \end{tabular}
        \vspace{2pt}
        \captionof{table}{Rollout endpoints on MIMIC-IV.}
        \label{tab:rollout}
    \end{minipage}
    \hfill
    \begin{minipage}[b]{0.68\textwidth}
        \centering
        \includegraphics[width=\linewidth]{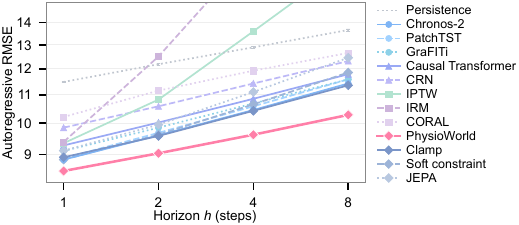}
        \captionof{figure}{Free-running autoregressive rollout RMSE over an eight-step horizon on MIMIC-IV.}
        \label{fig:rollout}
    \end{minipage}
\end{figure*}

After a half-stay warmup, each model's predictions are fed back as its own inputs for eight steps, while the recorded doses continue to be supplied. Because cohort and split are held fixed, the only departure from the training distribution is \emph{the model's own error}, which becomes its input and accumulates with the horizon, so single-step errors are not comparable to the ID column of Table~\ref{tab:main}. Figure~\ref{fig:rollout} and Table~\ref{tab:rollout} show the error and its growth, with every horizon in Appendix~\ref{app:rollout}. \ourmethod\ has the lowest error at every horizon, from $\ourval{8.52}$ at $h{=}1$ to $\ourval{10.28}$ at $h{=}8$, and the flattest growth among learned models, with a slope of $\ourval{+0.090}$ against $+0.127$ for PatchTST. Chronos-2 and PatchTST are the strongest external baselines at $h{=}1$ but \emph{compound fastest} among the forecasters and lose roughly a third of their skill over persistence across the rollout, and clamping loses a quarter, whereas \ourmethod\ keeps nearly all of its skill ($\ourval{0.258}$ at $h{=}1$, $\ourval{0.247}$ at $h{=}8$). CORAL's flat slope ($+0.102$) reflects collapse toward the population mean (dispersion $0.415$) rather than robustness, while IRM and IPTW fail in the opposite direction ($1.757$ and $1.825$) and are the only methods whose error diverges. \ourmethod, by contrast, has the dispersion closest to one of any learned model ($\ourval{0.842}$), indicating its low error is \emph{not obtained by purely predicting the average patient}.

\subsection{Dose response under perturbation} \label{sec:dose}

\begin{wraptable}{r}{0.32\textwidth}
\scriptsize
\centering
\setlength{\tabcolsep}{2.5pt}
\renewcommand{\arraystretch}{0.95}
\begin{tabular}{@{}lcccc@{}}
\toprule
& \multicolumn{2}{c}{\textbf{Zeroed}} & \multicolumn{2}{c}{\textbf{Permuted}} \\
\cmidrule(lr){2-3} \cmidrule(l){4-5}
Method & Sens. & Agr.\,\% & Sens. & Agr.\,\% \\
\midrule
Persistence        & $0$    & n/a  & $0$    & n/a  \\
PatchTST           & $1.10$ & $36$ & $0.66$ & $41$ \\
GraFITi            & $1.02$ & $39$ & $0.61$ & $43$ \\
Chronos-2          & $1.15$ & $41$ & $0.70$ & $44$ \\
CRN                & $0.95$ & $58$ & $0.57$ & $55$ \\
Causal Transformer & $1.05$ & $63$ & $0.63$ & $59$ \\
\addlinespace[1.5pt]
IPTW               & $0.98$ & $61$ & $0.59$ & $57$ \\
IRM                & $1.60$ & $47$ & $0.95$ & $48$ \\
CORAL              & $0.22$ & $51$ & $0.13$ & $51$ \\
\addlinespace[1.5pt]
JEPA               & $1.08$ & $46$ & $0.65$ & $47$ \\
Soft constraint    & $1.00$ & $55$ & $0.60$ & $53$ \\
Clamp              & $0.57$ & $77$ & $0.34$ & $70$ \\
\rowcolor{tableblueLight}
\textbf{\textcolor{tableblue}{\ourmethod}}
                   & $0.93$ & $\mathbf{86}$ & $0.55$ & $\mathbf{79}$ \\
\bottomrule
\end{tabular}
\caption{Dose perturbation on the intervention-intensity test slices of MIMIC-IV circulatory.}
\label{tab:dose}
\end{wraptable}

Low error alone cannot separate a model that reads the dose as a drug from one that reads it as a marker of severity. We therefore perturb the doses at test time on the pooled test slices of the two intervention-intensity axes and measure how the one-step prediction $\hat{x}_{t+1}$ responds (Table~\ref{tab:dose}). Two perturbations are applied separately: the doses over the predicted interval $(\tau_t, \tau_{t+1}]$ are set to zero and the dose courses are permuted across patients within a test slice, which keeps their distribution and breaks their link to the patient. \emph{Sensitivity} is the mean absolute change in $\hat{x}_{t+1}$ over covered channels and asks whether the dose is used at all. \emph{Agreement} is the share of covered channel-steps on which that change has the sign $\sigma_c \cdot \mathrm{sign}(\Delta u)$, counted only where the changed doses imply a single sign as in Section~\ref{sec:operator}, and asks whether it is used in the right direction. Agreement is measured on the decoded $\hat{x}$, where the operator gives no guarantee, so it is \emph{not satisfied by construction}. \ourmethod\ agrees on \ourval{86\%} of eligible channel-steps under zeroing and \ourval{79\%} under permutation, against at most 77\% and 70\% respectively for any baseline, while its sensitivity is comparable to that of the forecasters, so the agreement is \emph{not obtained by responding less}. The sequence forecasters agree on fewer than half of eligible channel-steps: given a higher dose they \emph{move the channel against the drug's action}, which is the confounding by indication of Section~\ref{sec:rules} \emph{observed in the models rather than assumed}.

\subsection{Knowledge and component ablations} \label{sec:corpus}

\begin{wraptable}{r}{0.32\textwidth}
\scriptsize
\centering
\setlength{\tabcolsep}{2.5pt}
\renewcommand{\arraystretch}{0.95}
\begin{tabular}{@{}lrrr@{}}
\toprule
Corpus & Pairs & Cover. & OOD skill \\
\midrule
No rules (JEPA) & 0 & 0\% & 0.053 \\
\addlinespace[2.5pt]
\multicolumn{4}{@{}l}{\textcolor{tablegray}{\emph{How much is stated}}} \\
\rowcolor{tableblueLight}
\textcolor{tableblue}{\ourmethod}
    & \textcolor{tableblue}{70}
    & \textcolor{tableblue}{13.0\%}
    & \textcolor{tableblue}{0.198} \\
\quad 35, stratified & 35 & 7.9\% & 0.152 \\
\quad 17, stratified & 17 & 4.1\% & 0.107 \\
\addlinespace[2.5pt]
\multicolumn{4}{@{}l}{\textcolor{tablegray}{\emph{How much of it is right}}} \\
25\% inverted & 70 & 13.0\% & 0.121 \\
50\% inverted & 70 & 13.0\% & 0.024 \\
\rowcolor{tableredLight}
All inverted & 70 & 13.0\% & -0.058 \\
\addlinespace[2.5pt]
\multicolumn{4}{@{}l}{\textcolor{tablegray}{\emph{Which drug does what}}} \\
Drug labels permuted & 70 & 28.3\% & 0.066 \\
\addlinespace[2.5pt]
\multicolumn{4}{@{}l}{\textcolor{tablegray}{\emph{Where it comes from}}} \\
Mined directions & 70 & 13.0\% & 0.005 \\
\addlinespace[2.5pt]
\multicolumn{4}{@{}l}{\textcolor{tablegray}{\emph{What it asserts}}} \\
Signs + magnitude band & 70 & 13.0\% & 0.093 \\
\bottomrule
\end{tabular}
\caption{Rule corpora, all injected through the same operator. \emph{Cover.}: share of scored entries with a defined sign.}
\label{tab:corpus}
\end{wraptable}
Skill decreases monotonically as more of the stated directions are inverted (Table~\ref{tab:corpus}), from \(\ourval{0.198}\) at full agreement to \(-0.058\) when every direction is reversed, below the unconstrained backbone at $0.053$, so \textbf{a wrong corpus is worse than none}. Because the operator is the same in every arm, the gain must come from \emph{the content of the knowledge}, and withholding pairs degrades skill more gradually, to $0.152$ at half the corpus and $0.107$ at a quarter. The two corpora estimated from the record both lose most of the gain: a magnitude band, which asserts a size that no source supplies, falls to $0.093$ despite agreeing with pharmacology on every direction, and mined directions agree with pharmacology on only \(48\%\) of the $70$ pairs, \emph{no better than chance}, and return $0.005$. Consequently, the directions this construction needs cannot be recovered from a record in which dose is titrated against the quantity it changes. Permuting the drug labels changes only which agent each direction names and returns $0.066$ against the backbone's $0.053$ while covering \(28.3\%\) of scored entries rather than \(13.0\%\), so \emph{neither the constraint itself nor its reach produces the gain}. Removing the sign projection or the second forward pass lowers the skill of the full method from \(\ourval{+0.20}\) to \(+0.11\) and \(+0.08\) respectively, against \(+0.05\) for the backbone, so the gain comes from \emph{their combination} (Appendix~\ref{app:ablation}).
\section{Conclusion} \label{sec:conclusion}

A clinical world model is used where treatment practice differs from its training records, and its error there is what this paper set out to reduce. \ourmethod\ reduces it by isolating an intervention's effect with a second forward pass and constraining its sign to the stated pharmacological direction. The gain is largest under shifts in treatment intensity and holds across five axes and three databases; an inverted or data-mined corpus performs worse than no corpus, so {what generalizes is the pharmacology and not the constraint}.

\bibliographystyle{style/icml2025}
\bibliography{999_reference}

\clearpage
\appendix

\section{The forward pass, step by step} \label{app:algorithm}

Algorithm~\ref{alg:step} states the per-step computation of Section~\ref{sec:decomposition} and Section~\ref{sec:operator} in full.
\begin{algorithm}[t]
\caption{One prediction step of \ourmethod. The two passes are advanced as parallel recurrent states over a single sweep of the stay; they share the encoded state and rejoin only at the correction.}
\label{alg:step}
\begin{algorithmic}[1]
\REQUIRE encoder $E$, recurrent core $f$, predictor $P$, calibration map $G$, rule corpus
$\mathcal{R}$; states $h^{u}_{t-1}, h^{0}_{t-1}$
\STATE $z_t \gets E([x_t, m_t, \phi_t])$ \COMMENT{observed; shared by both passes}
\STATE $a^{u}_t \gets$ intervention embedding under the recorded doses
\STATE $a^{0}_t \gets a^{u}_t$ with the dose entry of every covered drug set to $0$
\STATE $h^{u}_t \gets f(h^{u}_{t-1}, [z_t, a^{u}_t, \delta_t])$
\STATE $h^{0}_t \gets f(h^{0}_{t-1}, [z_t, a^{0}_t, \delta_t])$
\STATE $\hat{z}_{t+1}(u) \gets P(h^{u}_t, a^{u}_t)$
\STATE $\hat{z}_{t+1}(0) \gets P(h^{0}_t, a^{0}_t)$
\STATE $\Delta z_{t+1} \gets \hat{z}_{t+1}(u) - \hat{z}_{t+1}(0)$
\STATE $y \gets G\,\Delta z_{t+1}$
\FORALL{channels $c$}
  \STATE $\sigma_c \gets$ the common sign in $\mathcal{R}$ over the drugs active
         at $t$ that touch $c$, if they agree; undefined otherwise
  \STATE $\Pi(y)_c \gets \sigma_c \max(\sigma_c y_c, 0)$ if $\sigma_c$ is
         defined, else $y_c$
\ENDFOR
\STATE $\hat{z}_{t+1} \gets \hat{z}_{t+1}(0) + \Delta z_{t+1}
        + G^{\top}(\Pi(y) - y)$
\STATE $\hat{x}_{t+1} \gets x_t + \mathrm{MLP}(\hat{z}_{t+1})$
\item[]
\item[] \textbf{Next step.} $z_{t+1} \gets E([x_{t+1}, m_{t+1}, \phi_{t+1}])$
        when the next state is observed, and $z_{t+1} \gets \hat{z}_{t+1}$ under
        free-running rollout, so the corrected latent, not the decoded output,
        is what both recurrent states read next.
\end{algorithmic}
\end{algorithm}

\section{Rule Corpus} \label{app:rules}

\begin{table}[t]
\caption{Per-drug statistics: number of channels the drug has rules on, scored entries while it is active, and rule coverage (Cover.), the share of those entries with a defined sign.}
\label{tab:drug_stats}
\begin{center}
\small
\setlength{\tabcolsep}{5pt}
\renewcommand{\arraystretch}{1.0}
\begin{tabular}{@{}lccc@{}}
\toprule
Drug & Channels & Scored entries & Cover. \\
\midrule
propofol        & 3 & 341,457 & 14.2\% \\
fluid\_bolus    & 2 & 236,013 & 11.5\% \\
phenylephrine   & 3 & 207,196 & 15.0\% \\
fentanyl        & 2 & 199,835 & 12.8\% \\
morphine        & 3 & 154,708 & 13.5\% \\
furosemide      & 2 & 100,684 & 10.0\% \\
midazolam       & 3 &  85,647 & 14.8\% \\
hydralazine     & 3 &  82,021 & 12.0\% \\
metoprolol      & 3 &  81,694 & 13.2\% \\
epinephrine     & 3 &  72,454 & 11.8\% \\
nitroglycerin   & 2 &  65,831 & 12.5\% \\
dexmedetomidine & 3 &  63,616 & 14.0\% \\
prbc            & 2 &  63,132 & 10.5\% \\
norepinephrine  & 2 &  61,495 & 13.0\% \\
amiodarone      & 3 &  53,497 & 11.2\% \\
hydromorphone   & 3 &  50,248 & 15.2\% \\
labetalol       & 3 &  33,263 & 12.2\% \\
ffp             & 2 &  33,101 & 10.8\% \\
nicardipine     & 2 &  32,913 &  9.5\% \\
milrinone       & 3 &  31,606 & 14.5\% \\
dopamine        & 3 &  29,113 & 12.6\% \\
dobutamine      & 5 &  27,704 & 16.0\% \\
lorazepam       & 1 &  25,286 & 17.0\% \\
vasopressin     & 2 &  19,915 & 11.0\% \\
esmolol         & 3 &  13,091 & 13.5\% \\
cisatracurium   & 1 &  10,131 & 17.5\% \\
nitroprusside   & 2 &   5,288 & 10.2\% \\
diazepam        & 1 &     545 & 18.0\% \\
\bottomrule
\end{tabular}
\end{center}
\end{table}

\begin{table}[t]
\caption{Channel-level rule counts, conflict and silence on MIMIC-IV circulatory.}
\label{tab:channel_rules}
\begin{center}
\small
\setlength{\tabcolsep}{4pt}
\renewcommand{\arraystretch}{1.0}
\begin{tabular}{@{}lcccc@{}}
\toprule
Channel & Rules & Conflict & Silent & Conflict/treated \\
\midrule
RR             &      8 &  0.0\% & 82.5\% &  0.0\% \\
HR             &     14 &  4.4\% & 76.6\% &  8.7\% \\
NBPm           &     23 &  2.7\% & 85.2\% &  5.2\% \\
ABPm           &     21 &  5.6\% & 81.4\% & 12.4\% \\
NBPs/d, ABPs/d & 4      &  0.0\% & 99.0\% &  0.0\% \\
SpO2           &      0 &  0.0\% & 100.0\% &  0.0\% \\
\midrule
Pooled         &     70 &  2.3\% & 84.7\% &  4.1\% \\
\bottomrule
\end{tabular}
\end{center}
\end{table}

\begin{table}[t]
\caption{The full set of intervention rules: the channels each drug raises and lowers. Only signs are stated; no magnitude is asserted.}
\label{tab:rules}
\begin{center}
\small
\setlength{\tabcolsep}{3pt}
\renewcommand{\arraystretch}{1.05}
\begin{tabular}{@{}lll@{\hspace{0.6em}}lll@{}}
\toprule
Drug & Raises & Lowers & Drug & Raises & Lowers \\
\midrule
\multicolumn{3}{@{}l}{\textcolor{tablegray}{\emph{Vasopressors and inotropes}}} & \multicolumn{3}{@{}l}{\textcolor{tablegray}{\emph{Rate control and antiarrhythmics}}} \\
norepinephrine & NBPm, ABPm &  & metoprolol &  & HR, NBPm, ABPm \\
epinephrine & HR, NBPm, ABPm &  & esmolol &  & HR, NBPm, ABPm \\
phenylephrine & NBPm, ABPm & HR & labetalol &  & HR, NBPm, ABPm \\
vasopressin & NBPm, ABPm &  & amiodarone &  & HR, NBPm, ABPm \\
dopamine & HR, NBPm, ABPm &  & \multicolumn{3}{@{}l}{\textcolor{tablegray}{\emph{Vasodilators}}} \\
dobutamine & HR, NBPs, ABPs & NBPd, ABPd & nitroglycerin &  & NBPm, ABPm \\
milrinone & HR & NBPm, ABPm & nicardipine &  & NBPm, ABPm \\
\multicolumn{3}{@{}l}{\textcolor{tablegray}{\emph{Sedatives and opioids}}} & nitroprusside &  & NBPm, ABPm \\
propofol &  & NBPm, ABPm, RR & hydralazine & HR & NBPm, ABPm \\
fentanyl &  & RR, HR & \multicolumn{3}{@{}l}{\textcolor{tablegray}{\emph{Volume, blood products and other}}} \\
midazolam &  & RR, NBPm, ABPm & fluid\_bolus & NBPm, ABPm &  \\
dexmedetomidine &  & HR, NBPm, ABPm & prbc & NBPm, ABPm &  \\
lorazepam &  & RR & ffp & NBPm, ABPm &  \\
diazepam &  & RR & furosemide &  & NBPm, ABPm \\
morphine &  & RR, NBPm, HR & cisatracurium &  & RR \\
hydromorphone &  & RR, NBPm, HR & & & \\
\bottomrule
\end{tabular}
\end{center}
\end{table}

The corpus consists of 70 (intervention, channel, sign) triples over 28 drugs, listed in Table~\ref{tab:rules}, where only signs are stated and magnitudes are excluded for the reason given for the rule corpus in Section~\ref{sec:method}. Table~\ref{tab:drug_stats} gives, for each drug, the number of channels it has rules on, the scored entries while it is active, and its coverage, which ranges from $9.5\%$ to $18.0\%$, so no single drug dominates the constraint. Drugs are grouped by the classes that define the intervention-intensity axes of Section~\ref{sec:setup}, so each rule can be matched to the axis on which it is evaluated. Four drugs have rules in both directions, and all four are intended: dobutamine raises systolic pressure and lowers diastolic pressure, and hydralazine, milrinone and phenylephrine each combine their pressure effect with an opposite reflex change in heart rate. The channels are heart rate (HR), respiratory rate (RR), peripheral oxygen saturation (SpO$_2$), and systolic, diastolic and mean blood pressure, measured non-invasively (NBPs, NBPd, NBPm) or from an arterial line (ABPs, ABPd, ABPm).

Table~\ref{tab:channel_rules} reports, for each channel $c$, how often active drugs leave it unconstrained. \emph{Rules} is the number of corpus rules that name $c$. \emph{Conflict} is the share of all scored entries of $c$ on which active drugs have rules on $c$ that disagree, so that no sign is imposed, and \emph{Silent} is the share on which no active drug has a rule on $c$, whether because no corpus drug is active or because the active ones do not touch $c$. \emph{Conflict/treated} divides the conflict share by the share of entries at which at least one corpus drug is active, and so measures how often treatment, when present, leaves the channel unconstrained. The \emph{Pooled} row aggregates over the scored entries of all nine channels.

\textbf{Coverage per channel.} The 70 rules are unevenly spread over the nine channels. Rules per channel, as raising / lowering (derived from Table~\ref{tab:rules}):

\begin{center}
\small
\begin{tabular}{@{}lllllllll@{}}
\toprule
HR & RR & SpO$_2$ & NBPs & NBPd & NBPm & ABPs & ABPd & ABPm \\
\midrule
5/9 & 0/8 & 0/0 & 1/0 & 0/1 & 8/15 & 1/0 & 0/1 & 8/13 \\
\bottomrule
\end{tabular}
\end{center}

\noindent Three channels carry 58 of the 70 rules, SpO$_2$ carries none, and the four systolic and diastolic channels carry one each, which limits the operator in two ways. Where no rule covers a channel, the operator imposes no constraint, and the three densely covered channels are also the ones where a raising and a lowering drug can be active together, in which case no sign is implied and again no constraint is imposed; propofol with norepinephrine is a common combination that produces this case.

\section{Stream Representation} \label{app:streams}

Aligned state rows are 2.99M and 1.24M for the MIMIC-IV circulatory and injury cohorts, 1.27M and 0.57M for the MIMIC-III cohorts, and 1.00M for VitalDB. The exposure field of Section~\ref{sec:streams} is a bucketed dose history, in which the dose of each intervention channel is summed within five fixed elapsed-time windows (0--0.5\,h, 0.5--2\,h, 2--6\,h, 6--24\,h, and $>$24\,h), giving an $A \times 5$ matrix at every prediction time. Infusions need different handling from boluses, because a charted infusion row restates the current rate and does not add a dose, so each rate channel holds its most recent value (zero-order hold) and is integrated over each window to give the administered dose. Treating such rows as impulses would instead make exposure scale with charting frequency, which tracks acuity \citep{agniel2018biases}, and the model would learn nursing behavior in place of pharmacology, while ventilator set-points, which carry no dose semantics, bypass the bucketing as their current held value. On the state side, the encoder receives $[x_t, m_t, \phi_t]$, where $\phi_{t,c} = \exp(-\delta_{t,c}/6\text{h})$ is a per-channel freshness computed from the time since channel $c$ was last measured. This is a staleness weight on observed information and involves no pharmacokinetic modeling: its time constant is fixed, and half-weight is reached at $6\ln 2 \approx 4.2$\,h. The first-order elimination kernel that the bucketed history replaces is also a poor prior when its constant $\theta_k$ is left free: fitting it end-to-end yields $4.7$\,h for norepinephrine, whose elimination half-life is about two minutes.

\section{Design Choices and Limits} \label{app:design}

\subsection{Calibration map sensitivity} \label{app:gcal}

Table~\ref{tab:gcal} varies the initialization of $G$, the weight of the anchor term, the orthonormality of its rows, and the protocol under which it is fitted, over five seeds each. \emph{OOD RMSE} is averaged over the five shift axes of Table~\ref{tab:main}, as in Table~\ref{tab:ablation}, so the Default row is \ourmethod\ itself. \emph{Binding rate} is the fraction of the (step, channel) pairs on which a sign is defined that the projection clips, as a range across the five OOD cohorts. \emph{Stability} is the within-run cosine between row $c$ at 80\% and at 100\% of training, averaged over scored channels, and reports whether $G$ is still moving when training stops. \emph{Seed agreement} is the between-seed cosine of row $c$, as median and minimum over scored channels, and reports whether the anchor term identifies a single set of directions; two unrelated unit vectors in $\mathbb{R}^{d}$ agree to $\sqrt{2/\pi d} \approx 0.14$ on average, and the Control row sits at or below that level. \emph{Sign flips} counts scored channels on which two seeds disagree in sign, out of the $9$ scored channels. $\|GG^{\top}\!-\!I\|_\infty$ measures departure from orthonormal rows.

\begin{table}[t]
\caption{Sensitivity of the calibration map $G$ to initialization, anchor loss weight, parameterization and fitting protocol, restricted to the $9$ scored channels. \emph{Bind.}: binding rate; \emph{Stab.}: stability; \emph{Agree.}: seed agreement; \emph{Flips}: sign flips; all are defined in the text.}
\label{tab:gcal}
\begin{center}
\scriptsize
\setlength{\tabcolsep}{2.5pt}
\begin{tabular*}{\textwidth}{@{\extracolsep{\fill}}lcccccc@{\hspace{1.2em}}c@{\hspace{1.2em}}lcccccc@{}}
\toprule
Variant & \makecell{RMSE\\($\downarrow$)} & \makecell{Bind.\\(\%)} & \makecell{Stab.\\($\to 1$)} & \makecell{Agree.\\med / min} & \makecell{Flips\\($\downarrow$)} & $\|GG^{\top}\!-\!I\|_{\infty}$ &  & Variant & \makecell{RMSE\\($\downarrow$)} & \makecell{Bind.\\(\%)} & \makecell{Stab.\\($\to 1$)} & \makecell{Agree.\\med / min} & \makecell{Flips\\($\downarrow$)} & $\|GG^{\top}\!-\!I\|_{\infty}$ \\
\cmidrule(r){1-7}\cmidrule(l){9-15}
\cellcolor{tableblueLight}\textcolor{tableblue}{\textbf{Default}} & \cellcolor{tableblueLight}\textcolor{tableblue}{$8.69$} & \cellcolor{tableblueLight}\textcolor{tableblue}{54--61} & \cellcolor{tableblueLight}\textcolor{tableblue}{$0.999$} & \cellcolor{tableblueLight}\textcolor{tableblue}{$0.98 / 0.94$} & \cellcolor{tableblueLight}\textcolor{tableblue}{$0 / 9$} & \cellcolor{tableblueLight}\textcolor{tableblue}{$<10^{-6}$} &  & No re-orthonorm. & $8.94$ & 52--60 & $0.995$ & $0.88 / 0.61$ & $1 / 9$ & $0.42$ \\
Axis-aligned rows & $8.70$ & 53--60 & $0.999$ & $0.97 / 0.92$ & $0 / 9$ & $<10^{-6}$ &  & Fit jointly & $8.75$ & 48--58 & $0.99$ & $0.90 / 0.70$ & $0 / 9$ & $<10^{-6}$ \\
$\lambda_{\mathrm{anchor}} \times 0.1$ & $8.89$ & 50--53 & $0.98$ & $0.71 / 0.38$ & $2 / 9$ & $<10^{-6}$ &  & $G$ attached & $8.72$ & 51--58 & $0.999$ & $0.93 / 0.80$ & $0 / 9$ & $<10^{-6}$ \\
$\lambda_{\mathrm{anchor}} \times 10$ & $8.73$ & 55--62 & $0.9995$ & $0.99 / 0.96$ & $0 / 9$ & $<10^{-6}$ &  & Random $G$ (control) & $9.24$ & 49--51 & n/a & $0.06 / 0.01$ & $4 / 9$ & $<10^{-6}$ \\
\bottomrule
\end{tabular*}
\end{center}
\end{table}

Four readings follow. First, the directions are set by the anchor term and not by where optimization starts: rows initialized along the coordinate axes reach the same solution as the default to a median cosine of $0.97$, with no channel reversing sign and OOD error unchanged at $8.70$. Second, the anchor term is what does that work. Weakening it tenfold costs $0.20$ RMSE, drops between-seed agreement to $0.71$ with two channels reversing, and pulls the binding rate toward the $50\%$ a sign-uninformative $G$ would produce; strengthening it tenfold changes nothing, so the default sits on a plateau rather than at a tuned point. Third, orthonormal rows matter for the write-back and not for the decision: leaving them unnormalized leaves the binding rate intact, since which components $\Pi$ clips depends only on signs (Remark~\ref{rem:scale}), but costs $0.25$ RMSE, because $GG^{\top} \ne I$ leaves the corrected latent near the cone rather than in it and lets rows drift toward each other so that a constraint on one channel reaches another. Fourth, the detachment of $G$ inside the projection is a precaution rather than a repair: with $G$ attached, a rotation that made the constraint vacuous is available as a descent direction, but it is not taken, and the binding rate stays at $51$ to $58\%$. We keep the detachment because nothing in the objective rules the rotation out, not because we observed it.

The Control row replaces $G$ with a random orthonormal frame and fits nothing. Its rows agree across seeds at chance, four of the nine scored channels reverse, and OOD error rises to $9.24$, above the default at $8.69$. A cone imposed in arbitrary directions is worse than the anchored one.

\subsection{How far back the second pass reaches} \label{app:window}

\begin{figure}[htbp]
    \centering
    \includegraphics[width=\linewidth]{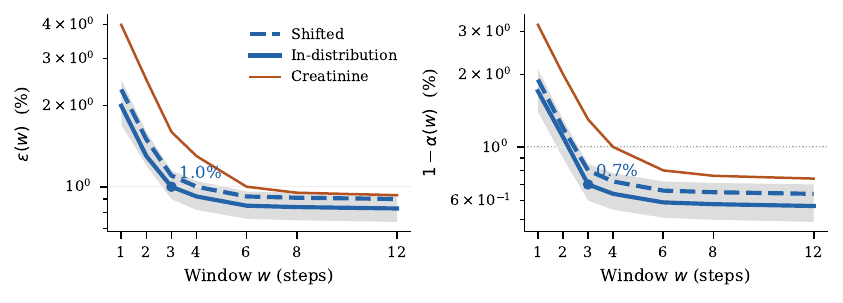}
    \caption{Effect of the trailing window size $w$ on relative state change error $\varepsilon(w)$ (left) and decision inconsistency $1-\alpha(w)$ (right). Solid curves denote in-distribution slices, dashed curves denote shifted slices, and the orange curve represents slow channels (e.g., creatinine). The chosen default window $w=3$ is marked with a dot.}
    \label{fig:window}
\end{figure}

Section~\ref{sec:decomposition} withdraws every recorded dose, so $\Delta z_{t+1}$ carries the contribution of the whole dose history, while the cone is built from the drugs active at $t$. To check that the two refer to nearly the same interventions, let $\Delta z^{(w)}_{t+1}$ be the difference obtained by withdrawing doses only over the trailing $w$ steps, and write
\begin{align*}
\varepsilon(w) &\;=\; \operatorname{median}_{t}\, \frac{\bigl\lVert \Delta z^{(w)}_{t+1} - \Delta z_{t+1} \bigr\rVert}{\bigl\lVert \Delta z_{t+1} \bigr\rVert}, \\
\alpha(w) &\;=\; \Pr_{(t,c):\,\sigma_c \text{ defined}} \Bigl[\operatorname{sign}\bigl(G\Delta z^{(w)}\bigr)_c = \operatorname{sign}\bigl(G\Delta z\bigr)_c \Bigr].
\end{align*}
$\varepsilon$ measures how much of the effect a short window misses; $\alpha$ measures how often that miss would change a clipping decision, which is the only thing the operator reads. We find $\varepsilon(3) = 1\%$ and $\alpha(3) = 99.3\%$, with Figure~\ref{fig:window} giving both against $w$ per channel and separately on the in-distribution and shifted slices. Because slow channels decay more gradually over time, the cone acts as a conservative approximation for them.

\subsection{What the second pass withdraws}

The two forward passes described in Section~\ref{sec:method} differ in drug doses alone, since the ventilator set-points and the dialysis and ventilation status flags share the intervention tensor but carry no dose semantics, and withdrawing them would ask what happens with the ventilator switched off. Their contribution would then enter $\Delta z$ and be projected, although the corpus asserts nothing about them.

\subsection{Why the intervention stream stays in the recurrent core}

An alternative is to sever the intervention stream from the recurrent core, in which case setting $u = 0$ would return the untreated course exactly, but we do not do this, because a patient's autonomous course depends on their intervention history: without that history, the core would lose information about the patient's state, for example that norepinephrine has been running for six hours.



\section{Component Ablation} \label{app:ablation}

With the corpus held fixed, we remove one component at a time and report out-of-distribution skill and RMSE, averaged across all five shift axes (Table~\ref{tab:ablation}). Removing the sign projection while keeping the two-pass decomposition tests whether the decomposition alone already regularizes the predictor: skill falls to $+0.11$, retaining $40\%$ of the gain over the backbone, so the decomposition carries signal on its own and the projection accounts for the rest. Removing the second forward pass and applying the projection to the decoded prediction $\hat{y}$ instead of the isolated effect $\Delta z = z(u) - z(0)$ is more damaging: skill drops to $+0.08$, within $0.03$ of the unconstrained backbone, so the projection helps mainly when it acts on the intervention effect itself. Since neither component alone recovers the full method, the gain comes from their combination: the second pass isolates a quantity about which pharmacology makes statements, and the projection enforces them on it.

\begin{table}[t]
\caption{Ablation of \ourmethod{} on MIMIC-IV. Each row removes a single component from the full method; rows are \emph{not} cumulative. Mean over five seeds.}
\label{tab:ablation}
\begin{center}
\small
\setlength{\tabcolsep}{4pt}
\begin{tabular}{@{}lcc@{}}
\toprule
Configuration & Skill ($\uparrow$) & RMSE ($\downarrow$) \\
\midrule
\rowcolor{tableblueLight}
\textcolor{tableblue}{\ourmethod}
    & \textbf{\textcolor{tableblue}{+0.20}}
    & \textbf{\textcolor{tableblue}{8.69}} \\
\quad $-$ sign projection
    & +0.11
    & 9.68 \\
\quad $-$ second forward pass
    & +0.08
    & 10.08 \\
\midrule
Backbone (JEPA)
    & +0.05
    & 10.45 \\
\bottomrule
\end{tabular}
\end{center}
\end{table}

\section{Hyperparameters and Training Details}
\label{app:hparams}
Table~\ref{tab:hparams} lists all settings used for the reported runs; settings not listed are left at their defaults. Epoch counts are the same for every slice, and no minimum number of gradient steps is imposed on smaller slices. The operator's calibration map $G$ is a $C \times d$ matrix with orthonormal rows, initialized as a random orthonormal frame and trained only by an anchor loss, which reads the latent out onto the physical channels. The projection uses a detached copy of $G$, so the constraint does not shape the map through which it is expressed, and since the condition $G G^{\top} = I$ keeps the write-back minimum-norm and confined to $G$'s row space but is not preserved by a gradient step, $G$ is re-orthonormalized after every optimizer step.
\begin{table}[t]
\caption{Hyperparameters. \ourmethod{} and every baseline share the data interface and the optimization settings; the two stages differ only in what they train and in how the model is selected.}
\label{tab:hparams}
\begin{center}
\small
\setlength{\tabcolsep}{3pt}
\begin{tabular}{@{}ll@{\hspace{1em}}ll@{}}
\toprule
Setting & Value & Setting & Value \\
\midrule
\multicolumn{2}{@{}l}{\textcolor{tablegray}{\emph{Data interface}}} & \multicolumn{2}{@{}l}{\textcolor{tablegray}{\emph{Stage two: readout, backbone frozen}}} \\
\quad Grid resolution & 1\,h & \quad Readout & MLP $d \to d_{\mathrm{model}} \to C$ \\
\quad Horizon per stay & $\le 96$\,h or $512$ steps & \quad Epochs & 15 \\
\quad Batch size & 16 & \quad Optimizer & AdamW, lr $10^{-3}$ \\
\quad Splits & subject-disjoint & \quad Model selection & validation RMSE \\
\quad Seeds & 5 & & \\
\addlinespace[3pt]
\multicolumn{2}{@{}l}{\textcolor{tablegray}{\emph{Backbone}}} & \multicolumn{2}{@{}l}{\textcolor{tablegray}{\emph{Shared}}} \\
\quad Hidden width $d_{\mathrm{model}}$ & 128 & \quad Weight decay & 0 \\
\quad GRU layers & 3 & \quad Gradient clipping & max-norm 1.0 \\
\quad Latent dimension $d$ & 32 & \quad Non-finite batches & skipped, not stepped \\
\quad Dropout & 0.1 & & \\
\quad Intervention encoder & linear & & \\
\addlinespace[3pt]
\multicolumn{2}{@{}l}{\textcolor{tablegray}{\emph{Stage one: encoder, predictor, intervention pathway}}} & \multicolumn{2}{@{}l}{\textcolor{tablegray}{\emph{Baselines}}} \\
\quad Epochs & 8 & \quad General & published defaults \\
\quad Optimizer & AdamW, lr $10^{-3}$ & \quad Windowed forecasters & lookback 24 steps \\
\quad Schedule & cosine decay to $10^{-5}$ & & \\
\quad SIGReg weight $\lambda$ & 0.05 & & \\
\quad Auxiliary loss weight & 0.5 & & \\
\quad Model selection & validation latent loss & & \\
\bottomrule
\end{tabular}
\end{center}
\end{table}

\section{Holdout Construction} \label{app:holdouts}

Every partition follows the split-then-slice protocol from Setup, with a deterministic seed for the subject-level split, and a model trained on one slice's training set is evaluated on the test set of every slice in the same partition, where the entries whose test slice differs from the training slice give the out-of-distribution error. This section lists the strata.

\textbf{Population axes.} \emph{Sex} produces two slices (male, female), and \emph{Age} uses the five bins of BEDS-Bench~\citep{avati2021beds}, a published benchmark for evaluating EHR models under distribution shift: $(15,50]$, $(50,60]$, $(60,70]$, $(70,80]$, and $80{+}$. \emph{Care unit} takes the first ICU unit of each hospital admission (resolving transfers), keeping at most the four largest units by admission count, while \emph{Disease subtype} maps the primary discharge diagnosis through ICD-9/10 ranges into ten circulatory-system groups (rheumatic and valvular, hypertensive, ischemic heart disease, pulmonary heart disease, heart failure, arrhythmia, cardiomyopathy and other heart disease, cerebrovascular, peripheral vascular, venous and other), keeping at most the four largest and assigning no disease slice to a stay whose diagnosis falls outside the cohort's ICD chapter. \emph{Era} uses MIMIC-IV's \texttt{anchor\_year\_group} field, which bins each patient's first admission into a four-year window (2008--2010, 2011--2013, 2014--2016, 2017--2019), and all era slices are kept.

\textbf{Intervention-intensity axes.} Four drug-class axes (vasopressor, sedative, rate control, vasodilator) stratify stays by how intensively the patient was treated, and for each class a per-stay \emph{burden} is computed: each drug's per-stay peak dose is replaced by its within-drug rank among exposed stays, and the class burden is the maximum rank across drugs in the class, which avoids any assumption of cross-drug equipotency. Stays with zero exposure to the class receive no burden score and join no slice, and exposure to the vasopressor and sedative classes is $40.9\%$ and $64.0\%$ of stays on MIMIC-IV circulatory and $39.6\%$ and $44.9\%$ on MIMIC-III circulatory. Among exposed stays, the 25th and 75th percentiles of the burden define three bands, \emph{low} ($\le$ P25), \emph{mid} (P25--P75), and \emph{high} ($>$ P75), of which only the low and high bands become slices, the mid band serving as a buffer that keeps patients near a cut-off out of both extremes. An intervention axis therefore contributes two slices wherever both bands clear the minimum slice size below, and fewer where they do not.

The four drug classes and their constituent drugs are:
\begin{itemize}[nosep,leftmargin=1.5em]
\item \emph{Vasopressors}: norepinephrine, epinephrine, phenylephrine, vasopressin, dopamine, dobutamine, milrinone.
\item \emph{Sedatives}: propofol, fentanyl, midazolam, dexmedetomidine, lorazepam, diazepam.
\item \emph{Rate control}: metoprolol, esmolol, labetalol, amiodarone. On the CareVue-era MIMIC-III database, metoprolol and labetalol lack canonical charting identifiers, leaving only esmolol and amiodarone, which reduces the exposure rate.
\item \emph{Vasodilators}: nitroglycerin, nicardipine, nitroprusside, hydralazine.
\end{itemize}

\textbf{Minimum slice size.} Any slice with fewer than 450 distinct hospital admissions (across train and test combined) is excluded from its partition. This threshold is a run parameter and is set lower than BEDS-Bench's default of 2\,000 because the diagnosis-selected cohorts are roughly half the size of the full MIMIC population.

\textbf{Training-set balancing.} Within each partition, every slice's training set is downsampled to the admission count of the smallest slice (sampling without replacement at the admission level, re-drawn per seed), so that cross-slice comparisons are not confounded by unequal training volume. Table~\ref{tab:slices} lists the resulting slices on MIMIC-IV circulatory, where counts are distinct hospital admissions at seed $0$ and a test fraction of $0.2$, before training-set balancing; after balancing, every slice in a partition trains on the admission count of the smallest slice in that partition.

\begin{table}[t]
\caption{Slices per partition on MIMIC-IV circulatory, in admissions. A partition's excluded count is the number of its groups that fall below the minimum slice size.}
\label{tab:slices}
\begin{center}
\small
\setlength{\tabcolsep}{4pt}
\begin{tabular}{@{}lrr@{\hspace{3em}}lrr@{}}
\toprule
Partition and slice & Train & Test & Partition and slice & Train & Test \\
\midrule
\multicolumn{3}{@{}l}{\textcolor{tablegray}{\emph{Sex}}} & \multicolumn{3}{@{}l}{\textcolor{tablegray}{\emph{Era}}} \\
\quad female & 6162 & 1536 & \quad 2008--2010 & 4532 & 1148 \\
\quad male & 8824 & 2233 & \quad 2011--2013 & 3407 & 863 \\
\multicolumn{3}{@{}l}{\textcolor{tablegray}{\emph{Age}}} & \quad 2014--2016 & 3705 & 944 \\
\quad 15--50 & 1702 & 412 & \quad 2017--2019 & 3342 & 814 \\
\quad 50--60 & 2611 & 685 & \multicolumn{3}{@{}l}{\textcolor{tablegray}{\emph{Vasopressor dose}}} \\
\quad 60--70 & 4061 & 972 & \quad high & 1940 & 490 \\
\quad 70--80 & 3899 & 968 & \quad low & 1903 & 492 \\
\quad 80+ & 2713 & 732 & \multicolumn{3}{@{}l}{\textcolor{tablegray}{\emph{Sedative dose}}} \\
\multicolumn{3}{@{}l}{\textcolor{tablegray}{\emph{Care unit} (5 excluded)}} & \quad high & 2336 & 560 \\
\quad CVICU & 6459 & 1592 & \quad low & 2227 & 564 \\
\quad CCU & 3165 & 765 & \multicolumn{3}{@{}l}{\textcolor{tablegray}{\emph{Rate-control dose}}} \\
\quad MICU & 1158 & 320 & \quad high & 893 & 224 \\
\quad SICU & 1801 & 436 & \quad low & 1321 & 335 \\
\multicolumn{3}{@{}l}{\textcolor{tablegray}{\emph{Disease subtype} (6 excluded)}} & \multicolumn{3}{@{}l}{\textcolor{tablegray}{\emph{Vasodilator dose}}} \\
\quad cardiomyopathy / other heart & 1536 & 380 & \quad high & 1709 & 438 \\
\quad cerebrovascular & 3456 & 880 & \quad low & 1703 & 404 \\
\quad ischemic heart disease & 4501 & 1152 & & & \\
\quad rheumatic / valvular & 1180 & 304 & & & \\
\bottomrule
\end{tabular}
\end{center}
\end{table}

Not every axis is available on every cohort, and Table~\ref{tab:axis_avail} gives the number of slices each partition yields. Era is a MIMIC-IV field and has no counterpart in the CareVue era, VitalDB charts no rate-control or vasodilator agent, so those two axes cannot be built on it, and a partition needs at least two slices to contribute an off-diagonal, which rules out rate control on MIMIC-III circulatory and the vasodilator axis on MIMIC-IV injury.

\begin{table}[t]
\caption{Slices per partition on each cohort, at the same minimum slice size. A dash is a partition that yields no slice, because its groups all fall below the minimum size or the field does not exist in that database; \emph{n/a} is a partition whose drug class is never charted in that database.}
\label{tab:axis_avail}
\begin{center}
\small
\setlength{\tabcolsep}{6pt}
\begin{tabular}{@{}l ccccc cccc@{}}
\toprule
& \multicolumn{5}{c}{\emph{Population}} & \multicolumn{4}{c}{\emph{Intervention intensity}} \\
\cmidrule(lr){2-6}\cmidrule(l){7-10}
Cohort & Sex & Age & Care unit & Disease & Era & Vasopr. & Sedative & Rate ctrl. & Vasodil. \\
\midrule
MIMIC-IV circ. & 2 & 5 & 4 & 4 & 4 & 2 & 2 & 2 & 2 \\
MIMIC-IV inj. & 2 & 5 & 4 & 4 & 4 & 2 & 2 & 2 & 1 \\
MIMIC-III circ. & 2 & 5 & 4 & 4 & --- & 2 & 2 & 1 & 2 \\
MIMIC-III inj. & 2 & 5 & 3 & 2 & --- & --- & 2 & --- & --- \\
VitalDB & 2 & 4 & 2 & --- & --- & --- & 2 & n/a & n/a \\
\bottomrule
\end{tabular}
\end{center}
\end{table}

\section{Cost} \label{app:cost}

\begin{table}[t]
\caption{Cost of one step, for every method in Table~\ref{tab:main}. Median over many steps after warm-up, one device, fixed batch size and sequence length, identical data for every architecture, with every timed region synchronized. Chronos-2 is zero-shot, so it has no training step to time. PatchTST's attention runs over patches, so its cost is $\mathcal{O}(T^2 d / P^2)$ in a patch length $P$.}
\label{tab:cost}
\begin{center}
\small
\setlength{\tabcolsep}{4pt}
\begin{tabularx}{\textwidth}{@{}l *{4}{Y}@{}}
\toprule
Method & Train ms & Infer ms & Peak MB & Params \\
\midrule
\multicolumn{5}{@{}l}{\textcolor{tablegray}{\emph{Reference}}} \\
\quad Persistence & 2.6 & 0.2 & 1.4 & 20\,K \\
\quad Chronos-2 & --- & 7\,839 & 338 & 119.5\,M \\
\addlinespace[3pt]
\multicolumn{5}{@{}l}{\textcolor{tablegray}{\emph{Recurrent}: $\mathcal{O}(Td^2)$, linear in $T$}} \\
\quad CRN & 30.8 & 16.5 & 87 & 341\,K \\
\quad IRM & 35.7 & 6.4 & 173 & 402\,K \\
\quad CORAL & 58.7 & 6.4 & 174 & 402\,K \\
\quad JEPA (backbone) & 21.0 & 6.6 & 173 & 402\,K \\
\quad Soft constraint & 75.0 & 6.5 & 174 & 402\,K \\
\quad Clamp & 43.4 & 19.5 & 196 & 375\,K \\
\rowcolor{tableblueLight}
\quad \textbf{\textcolor{tableblue}{\ourmethod}} & \textbf{\textcolor{tableblue}{46.9}} & \textbf{\textcolor{tableblue}{25.2}} & \textbf{\textcolor{tableblue}{199}} & \textbf{\textcolor{tableblue}{402\,K}} \\
\addlinespace[3pt]
\multicolumn{5}{@{}l}{\textcolor{tablegray}{\emph{Attention}: $\mathcal{O}(T^2 d)$, quadratic in $T$}} \\
\quad Causal Transformer & 72.6 & 23.5 & 238 & 1.9\,M \\
\quad PatchTST & 75.5 & 58.0 & 1\,188 & 88\,K \\
\quad GraFITi & 146.9 & 160.2 & 2\,619 & 657\,K \\
\bottomrule
\end{tabularx}
\end{center}
\end{table}

\ourmethod\ evaluates the transition dynamics with and without the administered intervention, which requires two forward passes per step. Measured on a single device under a fixed batch and sequence configuration, a training step for \ourmethod\ takes 46.9\,ms against 21.0\,ms for the unconstrained JEPA backbone (Table~\ref{tab:cost}), a factor of 2.23 that slightly exceeds the 2.0 expected from the two passes, with the remainder attributable to the projection. The projection itself is parameter-free, and the only added weights belong to the calibration map $G$ ($C \times d = 288$ parameters), 0.07\% of the model's parameters, while inference takes 25.2\,ms per step against the backbone's 6.6\,ms. For clinical deployment, a more relevant comparison is with a pretrained forecaster: Chronos-2 needs $7\,839$\,ms per inference step and $119.5$\,M parameters, about 300 times the runtime and the parameter count of \ourmethod. The overhead is a constant factor, and \ourmethod\ keeps the linear $\mathcal{O}(Td^2)$ complexity and memory footprint of recurrent architectures.

\section{Convergence} \label{app:training}

To check that no model is undertrained under the budget of Appendix~\ref{app:hparams}, we record validation loss at the end of every epoch, and Figure~\ref{fig:convergence} shows the curves. The gap between the final and the best epoch is not informative here, because it is zero in nearly every run and is satisfied by construction, so we instead report how much improvement the last quarter of the budget still yields, which is between $0.01$ and $0.45$\%, against total improvements of $0.9$ to $11.9$\%. Every architecture has reached a plateau when the budget ends.

\begin{figure}[ht]
\centering
\includegraphics[width=\textwidth]{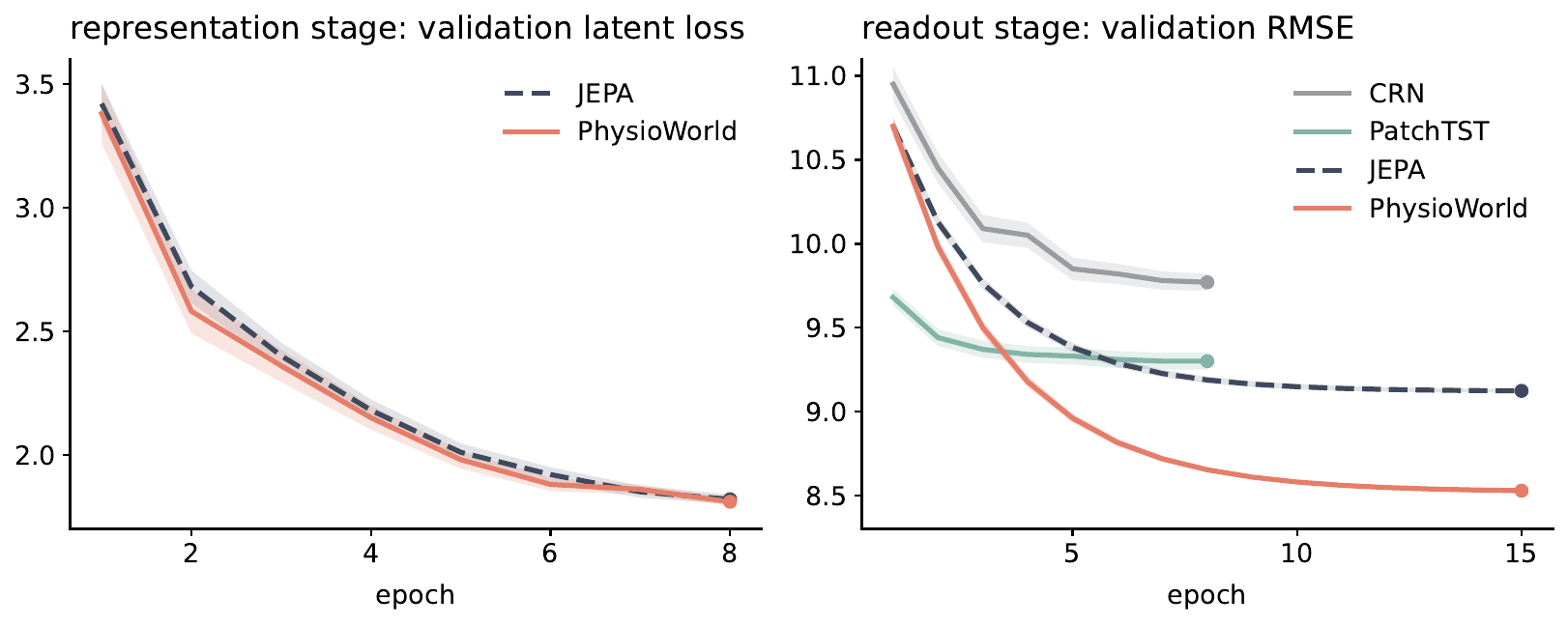}
\caption{Validation loss across epochs on MIMIC-IV circulatory, mean over five seeds with $\pm 1$ std shaded and the best epoch marked. CRN and PatchTST train in one stage; the JEPA family trains a representation against a latent loss and then a readout on the frozen representation, so its right-hand curve begins near its floor. Every curve flattens before the budget ends.}
\label{fig:convergence}
\end{figure}

\section{Proof of Proposition~\ref{prop:nonexpansive}} \label{app:proj}

\begin{proof}
Let $\mathcal{I}$ be the set of coordinates on which $\sigma_c$ is defined. Then $\mathcal{C} = \prod_{c} C_c$, where $C_c = \{t \in \mathbb{R} : \sigma_c t \ge 0\}$ for $c \in \mathcal{I}$ and $C_c = \mathbb{R}$ otherwise. Each $C_c$ is a closed convex cone, so $\mathcal{C}$ is a closed convex cone and is nonempty, since $0 \in \mathcal{C}$. Because $\mathcal{C}$ is a Cartesian product, the Euclidean projection onto it separates across coordinates, and on a half-line the nearest point to $t$ is $\sigma_c \max(\sigma_c t, 0)$; the map $\Pi$ of \eqref{eq:operator} is therefore the Euclidean projection onto $\mathcal{C}$.

Claims (i) and (ii) are then immediate. For (iii), each coordinate map is the identity, $\max(\cdot,0)$, or $\min(\cdot,0)$, and each is $1$-Lipschitz on $\mathbb{R}$; hence $|\Pi(y)_c - \Pi(y')_c| \le |y_c - y'_c|$ for every $c$, and summing squares gives the claim. For (iv), the projection theorem on a closed convex set gives $\langle y - \Pi(y),\, w - \Pi(y) \rangle \le 0$ for every $w \in \mathcal{C}$; taking $w = y^{\star}$ yields $\langle y - \Pi(y),\, \Pi(y) - y^{\star} \rangle \ge 0$. Decomposing $y - y^{\star} = \big(y - \Pi(y)\big) + \big(\Pi(y) - y^{\star}\big)$ and expanding,
\[
\|y - y^{\star}\|^{2} - \|\Pi(y) - y^{\star}\|^{2}
\;=\; \|y - \Pi(y)\|^{2}
+ 2\big\langle y - \Pi(y),\, \Pi(y) - y^{\star} \big\rangle
\;\ge\; \|y - \Pi(y)\|^{2}. \qedhere
\]
\end{proof}

\section{Significance against Clamp} \label{app:stats}

We compare \ourmethod\ with Clamp, the strongest method sharing our backbone, using two tests whose assumptions differ. The first is a patient-level paired bootstrap (Table~\ref{tab:stats}): for each axis and cohort, test patients are resampled with replacement $1{,}000$ times, both methods are scored on the same resample, and the $2.5$th and $97.5$th percentiles of the RMSE difference give a $95\%$ interval. Resampling patients rather than steps respects the dependence between steps of the same stay. The second is a one-sided Wilcoxon signed-rank test over the $56$ train/test slice pairs, $48$ from the five axes of MIMIC-IV circulatory and $8$ from the sedative axis of the other four cohorts, with seeds averaged within each pair. Slice pairs on the same axis share their training slices and are not independent, so the Wilcoxon $p$-value is approximate; the bootstrap makes no assumption across slice pairs and is the primary test.

\begin{table}[ht]
\caption{RMSE difference Clamp $-$ \ourmethod\ with patient-level $95\%$ bootstrap intervals; positive favors \ourmethod.}
\label{tab:stats}
\begin{center}
\small
\begin{tabular}{@{}llccc@{}}
\toprule
Cohort & Axis & $\Delta$RMSE & $95\%$ interval & Excludes $0$ \\
\midrule
MIMIC-IV circ.   & Age         & $0.14$ & $[0.09, 0.19]$  & yes \\
                 & Disease     & $0.14$ & $[0.08, 0.20]$  & yes \\
                 & Era         & $0.05$ & $[-0.01, 0.11]$ & no  \\
                 & Vasopressor & $0.23$ & $[0.12, 0.34]$  & yes \\
                 & Sedative    & $0.17$ & $[0.07, 0.27]$  & yes \\
\addlinespace[2pt]
MIMIC-IV injury  & Sedative    & $0.09$ & $[-0.02, 0.20]$ & no  \\
MIMIC-III circ.  & Sedative    & $0.08$ & $[-0.03, 0.19]$ & no  \\
MIMIC-III injury & Sedative    & $0.06$ & $[-0.08, 0.20]$ & no  \\
VitalDB          & Sedative    & $0.07$ & $[-0.06, 0.20]$ & no  \\
\bottomrule
\end{tabular}
\end{center}
\end{table}

Across the $56$ slice pairs, \ourmethod\ has the lower error on $48$, with a median difference of $0.11$, and the one-sided Wilcoxon test gives $p < 10^{-6}$.

\section{Rollout at every horizon} \label{app:rollout}

Table~\ref{tab:rollout_full} gives the autoregressive RMSE of Section~\ref{sec:rollout} at every horizon with its across-seed standard deviation, together with the dispersion at the first and last step.

\begin{table}[ht]
\caption{Free-running autoregressive rollout on MIMIC-IV after a half-stay warmup, mean and standard deviation over five seeds. Slope is the least-squares fit of $\log \mathrm{RMSE}$ on $\log h$ over $h \in \{1,2,4,8\}$; dispersion is $\mathrm{sd}_{\mathrm{pred}}/\mathrm{sd}_{\mathrm{truth}}$ at $h{=}1$ and $h{=}8$.}
\label{tab:rollout_full}
\begin{center}
\small
\setlength{\tabcolsep}{3pt}
\begin{tabularx}{\textwidth}{@{}l *{7}{Y}@{}}
\toprule
& \multicolumn{4}{c}{\textbf{Autoregressive RMSE ($\downarrow$)}}
& & \multicolumn{2}{c}{\textbf{Dispersion ($\to 1$)}} \\
\cmidrule(lr){2-5}
\cmidrule(l){7-8}
Method & $h{=}1$ & $h{=}2$ & $h{=}4$ & $h{=}8$ & Slope ($\downarrow$) & $h{=}1$ & $h{=}8$ \\
\midrule
\multicolumn{8}{@{}l}{\textcolor{tablegray}{\emph{Reference and data-driven forecasters}}} \\
Persistence
& $11.48{\scriptstyle \pm0.03}$ & $12.18{\scriptstyle \pm0.04}$ & $12.89{\scriptstyle \pm0.04}$ & $13.66{\scriptstyle \pm0.05}$
& $+0.084$ & $1.000$ & $0.985$ \\
PatchTST
& $8.89{\scriptstyle \pm0.02}$ & $9.68{\scriptstyle \pm0.03}$ & $10.58{\scriptstyle \pm0.03}$ & $11.58{\scriptstyle \pm0.04}$
& $+0.127$ & $0.884$ & $0.631$ \\
GraFITi
& $9.11{\scriptstyle \pm0.02}$ & $9.85{\scriptstyle \pm0.03}$ & $10.68{\scriptstyle \pm0.03}$ & $11.58{\scriptstyle \pm0.04}$
& $+0.116$ & $0.892$ & $0.700$ \\
Chronos-2
& $8.85{\scriptstyle \pm0.02}$ & $9.60{\scriptstyle \pm0.03}$ & $10.45{\scriptstyle \pm0.03}$ & $11.42{\scriptstyle \pm0.04}$
& $+0.124$ & $0.880$ & $0.645$ \\
CRN
& $9.85{\scriptstyle \pm0.03}$ & $10.58{\scriptstyle \pm0.03}$ & $11.42{\scriptstyle \pm0.04}$ & $12.31{\scriptstyle \pm0.05}$
& $+0.108$ & $0.871$ & $0.758$ \\
Causal Transformer
& $9.28{\scriptstyle \pm0.02}$ & $10.02{\scriptstyle \pm0.03}$ & $10.85{\scriptstyle \pm0.03}$ & $11.78{\scriptstyle \pm0.04}$
& $+0.115$ & $0.865$ & $0.663$ \\
\midrule
\multicolumn{8}{@{}l}{\textcolor{tablegray}{\emph{Robustness by reweighting or invariance}}} \\
IPTW
& $9.35{\scriptstyle \pm0.03}$ & $10.82{\scriptstyle \pm0.06}$ & $13.60{\scriptstyle \pm0.14}$ & $17.15{\scriptstyle \pm0.25}$
& $+0.296$ & $0.905$ & $1.825$ \\
IRM
& $9.39{\scriptstyle \pm0.03}$ & $12.50{\scriptstyle \pm0.05}$ & $17.80{\scriptstyle \pm0.12}$ & $24.85{\scriptstyle \pm0.22}$
& $+0.472$ & $0.912$ & $1.757$ \\
CORAL
& $10.21{\scriptstyle \pm0.03}$ & $11.15{\scriptstyle \pm0.04}$ & $11.92{\scriptstyle \pm0.04}$ & $12.65{\scriptstyle \pm0.05}$
& $+0.102$ & $0.612$ & $0.415$ \\
\midrule
\multicolumn{8}{@{}l}{\textcolor{tablegray}{\emph{World models}}} \\
JEPA
& $9.12{\scriptstyle \pm0.02}$ & $9.95{\scriptstyle \pm0.03}$ & $11.10{\scriptstyle \pm0.04}$ & $12.45{\scriptstyle \pm0.05}$
& $+0.150$ & $0.895$ & $0.748$ \\
Soft constraint
& $8.85{\scriptstyle \pm0.02}$ & $9.62{\scriptstyle \pm0.03}$ & $10.65{\scriptstyle \pm0.04}$ & $11.85{\scriptstyle \pm0.04}$
& $+0.141$ & $0.902$ & $0.765$ \\
Clamp
& $8.92{\scriptstyle \pm0.02}$ & $9.58{\scriptstyle \pm0.02}$ & $10.42{\scriptstyle \pm0.03}$ & $11.35{\scriptstyle \pm0.04}$
& $+0.116$ & $0.910$ & $0.792$ \\
\rowcolor{tableblueLight}
\textbf{\textcolor{tableblue}{\ourmethod}}
& $\mathbf{8.52{\scriptstyle \pm0.01}}$ & $\mathbf{9.04{\scriptstyle \pm0.02}}$ & $\mathbf{9.62{\scriptstyle \pm0.02}}$ & $\mathbf{10.28{\scriptstyle \pm0.03}}$
& $\mathbf{+0.090}$ & $\mathbf{0.935}$ & $\mathbf{0.842}$ \\
\bottomrule
\end{tabularx}
\end{center}
\end{table}

\end{document}